\documentclass[onecolumn,12pt,draftclsnofoot]{IEEEtran}

\usepackage{hyperref}

\usepackage[utf8]{inputenc}
\usepackage[T1]{fontenc}
\usepackage{cite}
\usepackage{url}
\usepackage[cmex10]{amsmath}
\usepackage{amssymb}
\usepackage{bbm}
\usepackage{mathrsfs}
\usepackage{graphicx}
\usepackage{caption}
\usepackage{subcaption}
\usepackage{algorithm,algpseudocode,tabularx,diagbox}
\usepackage{tikz}
\usetikzlibrary{shapes.geometric, arrows.meta, positioning, calc}

\usepackage{amsthm}
\newtheorem{theorem}{Theorem}[section]
\newtheorem{lemma}[theorem]{Lemma}
\newtheorem{definition}[theorem]{Definition}
\newtheorem{remark}[theorem]{Remark}

\title{Covert Bayesian Quickest Change Detection}

\author{
Yun-Feng Lo and Matthieu R. Bloch
\thanks{Y.-F. Lo and M. R. Bloch are with the School of Electrical and Computer Engineering, Georgia Institute of Technology, Atlanta, GA 30332 USA. Email: yun-feng.lo@gatech.edu, matthieu.bloch@ece.gatech.edu.}
}

\begin{document}
\maketitle

\begin{abstract}
We investigate the problem of covert quickest change detection in a Bayesian and infinite-horizon setting. A legitimate entity seeks to detect a change in the state of a discrete memoryless channel as quickly as possible by actively probing it. Simultaneously, the entity must ensure its probing remains covert from an adversary monitoring the channel for active sensing. We introduce the expected covertness budget (ECB) as an analytically tractable covertness metric that bounds from above the relative entropy between the observation sequences induced by active and passive sensing. Under constraints on both the probability of false alarm (PFA) and the ECB, we establish a second-order asymptotic converse bound on the average detection delay as the PFA constraint approaches zero, for any positive ECB constraint, explicitly quantifying the maximum square-root-order covert sensing gain possible. Furthermore, we propose an achievability scheme utilizing a constant-sensing-probability Shiryaev-type policy and show that it matches the second-order asymptotic converse. We illustrate our result with a numerical example.
\end{abstract}

\section{Introduction}
\label{sec:intro}
The study of the estimation-theoretic and information-theoretic limits of information hiding has found many applications, from steganography~\cite{Moulin2003, Ker2007, Wang2008} to covert communications~\cite{Bash2013, Bloch2015b, Wang2016b}, covert sensing in classical~\cite{Goeckel2017, Tahmasbi2020b, Wang2024Covert} and quantum~\cite{Bash2017, Gagatsos2019, Tahmasbi2020c, Tahmasbi2021, Hao2022Demonstration} systems, and to understanding the ability to perform stealthy attacks in cyber-physical systems~\cite{Bai2014, Bai2015}.

Particularly relevant to the present work is the literature on covert sensing, in which legitimate parties attempt to sense a system and estimate a parameter without triggering detection by an adversary monitoring whether the system is being probed. Much of the literature on covert sensing has focused on the problem of covert state estimation~\cite{Goeckel2017, Bash2017, Gagatsos2019, Hao2022Demonstration} and covert hypothesis testing~\cite{Tahmasbi2020b, Tahmasbi2020c, Tahmasbi2021,chang_covert_2021}. Therein, the objective is to characterize how a covert constraint affects the optimal mean-squared error or probability of discrimination error, and how these quantities scale with the number of observations.

We are interested here in studying the distinct problem of covert quickest change detection, in which the objective is for a terminal to probe a system that will change from one state to another at an a priori unknown time and minimize the average change detection delay (see, e.g.,~\cite{poor_quickest_2008,veeravalli_quickest_2014,tartakovsky_sequential_2014,xie2021sequential}, for introduction, exposition, and survey on the research field of quickest change detection). The problem is related to quickest change detection with controlled sensing~\cite{Veeravalli_2024}, with the additional requirement that probing remains covert from an adversary. 

There is related literature on covert quickest change detection, albeit with distinct models and objectives that result in rather different results. In~\cite{Huang_2020, Huang_2021}, the authors study a covert communication problem in which an adversary attempts to detect the transmission by running a quickest change detection algorithm. In~\cite{Ramtin_2024, Ramtin_2025}, the authors introduce the problem of quickest change detection in the presence of a covert adversary that attempts to do maximum damage by controlling the changepoint and post-change distribution, where maximum damage means making average detection delay on the order of the average time to false alarm. The present work differs in that a legitimate party attempts to perform quickest change detection while avoiding detection by an adversary monitoring the presence of probing signals.

It is worth noting that, although the policy optimization aspect of the covert Bayesian quickest change detection problem we define is somewhat similar to the problem setup investigated in~\cite{premkumar_optimal_2008, banerjee_bayesian_2011, geng_bayesian_2014}, in the sense that all can be formulated as partially observable Markov decision processes (POMDPs), in our work the use of a policy with vanishing probability of sensing, imposed by the covertness constraint, presents unique challenges.

Our problem is also related to that of~\cite{rao_extensible_2022}, both of which aim to ensure covertness in a sequential, infinite-horizon setting. It is interesting to observe that while the authors of~\cite{rao_extensible_2022} propose the use of a power sequence that strictly decays over time, our proposed policy uses a constant-over-time sensing probability, the value of which decays with the absolute log probability of false alarm, which serves as an ``effective blocklength" for the covert sensing problem.

An intuitive explanation of our main result is as follows. While the covertness constraint prevents active probing from improving the first-order scaling of the average detection delay, cautious probing still yields a second-order gain. In particular, the gain is of square-root order in the absolute logarithm of the false-alarm probability, reflecting a square-root-law phenomenon analogous to that in covert communication.

Our main contributions are summarized as follows.
First, we formulate covert Bayesian quickest change detection as a sequential active sensing problem constrained by both Alice's probability of false alarm and Eve's ability to detect probing. We introduce the expected covertness budget (ECB), a tractable surrogate that upper bounds Eve's relative-entropy-based distinguishability and remains amenable to policy design.

Second, we prove a non-asymptotic information bound that separates passive information from the geometric prior and active information from post-change sensing. Together with the ECB constraint, this yields a second-order converse: covert probing cannot improve the first-order average-delay scaling, but it can provide a square-root-order gain.

Third, we construct a constant-sensing-probability Shiryaev-type policy that satisfies the ECB constraint and matches the converse up to second order. Numerical experiments illustrate that the resulting covert probing gains are already visible at moderate false-alarm levels.

The remainder of the paper is organized as follows. In Sec.~\ref{sec:notation}, we collect our notation. In Sec.~\ref{sec:sys_model}, we define the problem of covert Bayesian quickest change detection. In Sec.~\ref{sec:covert_constr_analys}, we analyze the covertness constraint. In Sec.~\ref{sec:opt_prob}, we define Alice's active policy as an optimization problem. In Sec.~\ref{sec:sec-ord-conv}, we prove the second-order converse, while in Sec.~\ref{sec:achievability} we construct an achievability scheme matching that second-order converse. Finally, in Sec.~\ref{sec:numer_resul}, we present numerical results corroborating our findings.

\section{Notation}
\label{sec:notation}
We collect here the notation used throughout this paper.

We use $\triangleq$ to signify that the equality is a definition or follows directly from a definition.
Let $\mathbb{N}\triangleq\{1,2,\cdots\}$ denote the set of strictly positive integers.
Let $\emptyset$ denote the empty set or empty sequence, which will be clear from the context.
Let $\mathbb{R}$ denote the set of real numbers. For $x\in\mathbb{R}$, let $(x)^+\triangleq\max\{x,0\}$. For $a,b\in\mathbb{R}$ with $a<b$, let $[a,b]\triangleq\{r\in\mathbb{R}:a\le r\le b\}$, $]a,b[~\triangleq\{r\in\mathbb{R}:a< r< b\}$ and $]a,b]~\triangleq\{r\in\mathbb{R}:a< r\le b\}$. Let $\inf$ denote infimum, with the convention $\inf\emptyset\triangleq\infty$.
For any discrete set $\mathcal{S}$, let $\Delta_{\mathcal{S}}$ denote the probability simplex over $\mathcal{S}$, i.e., the set of probability mass functions (PMFs) over the alphabet $\mathcal{S}$. 
For any two sets $\mathcal{S}_1$ and $\mathcal{S}_2$, let $\mathcal{S}_1\times\mathcal{S}_2$ denote their Cartesian product. For any $n\in\mathbb{N}$ and for any set $\mathcal{S}$, we denote by $\mathcal{S}^n$ the $n$-fold Cartesian product of $\mathcal{S}$ with itself.
When a random variable (r.v.) $X$ follows a distribution $P$, we write $X\sim P$. In particular, $\textnormal{Ber}(\beta)$ stands for the Bernoulli distribution with success probability $\beta\in[0,1],$ and $\textnormal{Geo}(\rho)$ denotes the geometric distribution with parameter $\rho\in~]0,1[~$. 
For a sequence of r.v.s $(X_n)_{n\in\mathbb{N}}=(X_1,X_2,\cdots)$ or its realizations $(x_n)_{n\in\mathbb{N}}=(x_1,x_2,\cdots)$, let $X_n^m\triangleq(X_n,\cdots,X_m)$ if $n,m\in\mathbb{N}$ with $m\ge n$, otherwise $X_n^m=\emptyset$. 
We use the shorthand $X^m \triangleq X^m_1$ for $m\in\mathbb{N}$.
The same rules apply to the realizations.
By convention, a summation $\sum_{i=n}^m(\cdot)\triangleq0$ whenever $n>m$.
Let $\mathbb{I}\{\cdot\}$ denote the indicator function, and $\ln$ denotes natural logarithm.
For two PMFs $P$ and $Q$, we write $P\ll Q$ to denote that $P$ is absolutely continuous w.r.t. $Q$. For two PMFs $P$ and $Q$ with $P\ll Q$, let $\mathbb{D}(P\|Q)\triangleq\mathbb{E}_{X\sim P}\left[\ln\frac{P(X)}{Q(X)}\right]$, the relative entropy from $P$ to $Q$, and $\mathbb{V}(P\|Q)\triangleq\mathbb{E}_{X\sim P}\left[\left(\ln\frac{P(X)}{Q(X)}\right)^2\right]$. For $n\in\mathbb{N}$, define $\chi_{n+1}(P\|Q)\triangleq\sum_{z\in\mathcal{Z}} (P(z)-Q(z))^{n+1}/Q(z)^{n}$ for any PMFs $P,~Q$  with $P\ll Q$ on a discrete support $\mathcal{Z}$.

\section{System Model}
\label{sec:sys_model}

\begin{figure}[htbp]
    \centering

    \begin{tikzpicture}[
        >=Stealth,
        auto,
        node distance=1cm,
        block/.style={rectangle, draw, thick, minimum width=1.5cm, minimum height=1cm, align=center},
        channel/.style={rectangle, draw, thick, minimum width=1.5cm, minimum height=1cm, align=center},
        delay/.style={rectangle, draw, thick, minimum width=1.5cm, minimum height=1cm, align=center}
    ]
    
        \node[block] (alice) {Alice\\$\Pi^\textnormal{A}$};
        \node[channel, right=1cm of alice] (channel) {Channel\\$W_{Y,Z|X,\theta}$};
        \node[block, right=1cm of channel] (eve) {Eve\\$\mathscr{D}$};
        
        \node[block, above=1cm of channel] (nature) {Nature\\$\Gamma$};
        \draw[->, thick] (nature.south) -- node[right] {$\theta_t$} (channel.north);
        
        \draw[->, thick] (alice.east) -- node[above] {$X_t$} (channel.west);
        \draw[->, thick] (channel.east) -- node[above] {$Z_t$} (eve.west);
        
        \node[delay, below=1cm of channel] (delay) {Delay\\$z^{-1}$};
        
        \draw[->, thick] (channel.south) -- node[right] {$Y_t$} (delay.north);
        \draw[->, thick] (delay.west) -| (alice.south);
    
    \end{tikzpicture}

    \caption{Model for covert Bayesian quickest change detection.}
    \label{fig:system_model}
\end{figure}
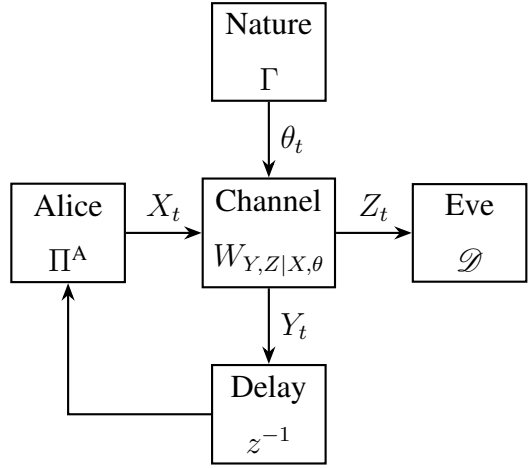

We consider the scenario illustrated in Figure~\ref{fig:system_model} in which there is a discrete memoryless channel (DMC) $W_{Y, Z|X,\theta}$ between Alice and Eve, where a third-party we call Nature controls the channel state sequence $(\theta_t)_{t\in\mathbb{N}}$. Alice, the legitimate party, controls the channel input sequence $(X_t)_{t\in\mathbb{N}}$ and receives one of the channel output sequences $(Y_t)_{t\in\mathbb{N}}$ after $1$ unit of time delay (denoted by the delay block labeled $z^{-1}$), while Eve, the adversary/eavesdropper, obtains the other channel output sequence $(Z_t)_{t\in\mathbb{N}}$ without delay. We assume that for all $t\in\mathbb{N}$, $X_t\in\mathcal{X}$, $Y_t\in\mathcal{Y}$ and $Z_t\in\mathcal{Z}$, where $\mathcal{X}$, $\mathcal{Y}$ and $\mathcal{Z}$ are the respective alphabets. For simplicity, we assume all these alphabets are discrete; in particular, $\mathcal{X}\triangleq\{0,1\}$, where $0$ denotes an innocent action and $1$ an active probing action.

We assume that the channel has only two states, denoted by $0$ and $1$, i.e., $\theta\in\Theta\triangleq\{0,1\}$. While, in principle, Nature can vary the state arbitrarily, in this paper, we consider the quickest change detection setup. That is, at some time instant $\Gamma\in\mathbb{N}$ that we call the changepoint, Nature switches the channel state $\theta$ from $0$ to $1$. Thus, the channel state sequence is
\begin{equation}
    \label{eq:channel_state_sequence}
    \theta_t
    =
    \begin{cases}
        0,&t\leq\Gamma,\\
        1,&t>\Gamma.
    \end{cases}
\end{equation}
In terms of indicator functions, we have $\theta_t=\mathbb{I}\{t>\Gamma\}$. We also denote $\theta_t(k)\triangleq\mathbb{I}\{t>k\}$ when the changepoint realization is $k$.
Note that our indexing convention differs by one time unit from the common Bayesian quickest change detection convention used in, e.g.,~\cite{tartakovsky_general_2005}, where the observation at the changepoint is drawn from the post-change distribution. Specifically, if $\lambda$
denotes the changepoint under that convention, then our convention corresponds to $\Gamma = \lambda - 1$. 

We further assume that this changepoint has a prior distribution that is geometric, following the classical Bayesian setting by Shiryaev~\cite{shiryaev1963optimum}. That is, the changepoint r.v. $\Gamma$ has the PMF
\begin{equation}
    \label{eq:prior_pmf}
    \mathbb{P}\{ \Gamma = k\} \triangleq \pi_k,
    ~\forall k\in\mathbb{N},
\end{equation}
where
\begin{equation}
    \label{eq:prior_mass_sequence}
    \pi_k\triangleq\rho(1-\rho)^{k-1}
\end{equation}
for some parameter $\rho\in~]0,1[$. We assume that both Alice and Eve know this prior (including the value of $\rho$), but neither of them knows the realization of the changepoint.
Let $d\triangleq |\ln(1-\rho)|>0$ be the exponent of the prior tail $T_n\triangleq\mathbb{P}_{\pi}\{\Gamma>n\}$, i.e., for $n\in\mathbb{N}$,
\begin{align}
    \label{eq:geo_prior_tail}
    T_n = \sum_{k>n} \pi_k = (1-\rho)^n = \exp(-dn)
    .
\end{align}

Alice aims to deploy a quickest change detection policy $\Pi^\textnormal{A}$ to detect the channel state change as quickly as possible, measured by an average detection delay (ADD) metric we define later. However, to prevent Alice from being too careless in declaring that a change has occurred, another metric also constrains her, the probability of false alarm (PFA), defined below. To gather data to support the hypothesis that the change has happened, Alice designs her probing action sequence $(X_t)_{t\in\mathbb{N}}$ and observes one of the channel output sequences $(Y_t)_{t\in\mathbb{N}}$ with a unit delay. By this we mean that, for all $t\in\mathbb{N}$, before deciding her next probing action $X_t$, Alice gets the delayed observation $Y_{t-1}$ in time to make a generally causal and closed-loop control (which we call causal control hereafter) action $X_{t}=s_t(X^{t-1}, Y^{t-1})$, where $s_t:\mathcal{X}^{t-1}\times\mathcal{Y}^{t-1}\to\Delta_{\mathcal{X}}$ is a random function of past actions $X^{t-1}$ and past observations $Y^{t-1}$. Since $\mathcal{X}\triangleq\{0,1\}$, an equivalent description of this causal control relationship is
\begin{equation}
    \label{eq:causal_control}
    X_t \sim \textnormal{Ber} \left( \beta_t(X^{t-1},Y^{t-1}) \right)
    ,
\end{equation}
where $\beta_t:\mathcal{X}^{t-1}\times\mathcal{Y}^{t-1}\to[0,1]$ is a deterministic function. 

For any $t\in\mathbb{N}$, let $\mathcal{F}_t\triangleq\sigma(X^{t-1},Y^{t-1})$ be the sigma algebra generated by the r.v.s $(X^{t-1},Y^{t-1})$, representing the causal information Alice has at time $t$, where $\mathcal{F}_1$ is the trivial sigma algebra. Note that because of the unit delay, this is different from the conventional definition where $\mathcal{F}_t=\sigma(X^{t}, Y^{t})$. Then the sequence $(\mathcal{F}_t)_{t\in\mathbb{N}}$ is a filtration that we call Alice's natural filtration. With this definition, $(\beta_t)_{t\in\mathbb{N}}$, viewed as a random process with randomness induced by its input arguments, is measurable with respect to (w.r.t.) $\mathcal{F}_{t}$. We call $\boldsymbol{\beta}\triangleq(\beta_t)_{t\in\mathbb{N}}$ the active sensing part of Alice's policy $\Pi^\textnormal{A}$.

Alice also uses a sequential test for the channel parameter change in the form of a stopping time $\tau^\textnormal{A}$, namely, for any $t\in\mathbb{N}$, the event $\{\tau^\textnormal{A}\leq t\}\in\mathcal{F}_t$. That is, for any realization $s$ of the r.v. $\tau^\textnormal{A}$, Alice believes that the change has happened at time $s$, solely based on the causal information available to her at that time, $\mathcal{F}_s$. Equivalently, Alice can deploy a sequence of tests $(d_t)_{t\in\mathbb{N}}$ where each $d_t:\mathcal{X}^{t-1}\times\mathcal{Y}^{t-1}\to\{0,1\}$ is a deterministic function. When $d_t(X^{t-1}, Y^{t-1})=0$, Alice continues running her policy and samples the next probing action $X_t$. When $d_t(X^{t-1}, Y^{t-1})=1$, Alice stops and believes that the change has happened. In this case, she keeps sending the innocent action $0$ onwards. We assume that $d_1\equiv0$, meaning that Alice does not stop without any gathered data.
The equivalence of the stopping time $\tau^\textnormal{A}$ and the sequence of tests $(d_t)_{t\in\mathbb{N}}$ can be seen via the following relation:
\begin{equation}
    \label{eq:stop_time_equiv}
    \tau^\textnormal{A} = \inf\{ t \in \mathbb{N}: d_t(X^{t-1},Y^{t-1})=1 \}.
\end{equation}
We refer to $\tau^\textnormal{A}$ as the sequential test part of Alice's policy $\Pi^\textnormal{A}$.

To sum up, Alice's policy can be represented as the tuple $\Pi^\textnormal{A}=(\tau^\textnormal{A},\boldsymbol{\beta})$. In particular, at time $t\in\mathbb{N}$, Alice first gathers the delayed observation feedback $Y_{t-1}$, and then decides whether to stop based on $\tau^\textnormal{A}$. If she decides to continue, she samples the probing action $X_t$ based on $\boldsymbol{\beta}$ and sends it to the channel input port. A particular policy of Alice's, which we call the ``innocent policy" and denote it by $\Pi^{\boldsymbol{0}}$, is to set $\boldsymbol{\beta}\triangleq\boldsymbol{0}$, i.e., setting all $\beta_t$ to zero and thus, Alice always sends the innocent action $X_t=0$ for all $t\in\mathbb{N}$. We also assume that Alice is optimistic in the sense that she has a policy of her choosing that is not the innocent policy. We call this Alice's active policy, and denote it $\Pi^\textnormal{A}$, a slight abuse of notation.

Eve represents an eavesdropper that is trying to decide whether Alice is innocent (i.e., using the innocent policy) or not, by observing one of the channel output sequences $(Z_t)_{t\in\mathbb{N}}\triangleq\boldsymbol{Z}$. In particular, she deploys a detector $\mathscr{D}$, which is a hypothesis testing protocol. For simplicity, we assume that Eve has full knowledge of Alice's active policy $\Pi^\textnormal{A}$. Then Eve's detector is a binary hypothesis test between the null hypothesis $H_0$ and the alternative hypothesis $H_1$ based on her observation, where
\begin{align}
    \label{eq:eve_hypo_test_inf_length}
    \begin{split}
    &H_0:~\boldsymbol{Z}\sim \mathbb{P}^{\Pi^{\boldsymbol{0}}}_{\boldsymbol{Z}},\\
    &H_1:~\boldsymbol{Z}\sim \mathbb{P}^{\Pi^{\textnormal{A}}}_{\boldsymbol{Z}},
    \end{split}
\end{align}
and $\mathbb{P}^{\Pi^{\boldsymbol{0}}}_{\boldsymbol{Z}}$ (resp. $\mathbb{P}^{\Pi^{\textnormal{A}}}_{\boldsymbol{Z}}$) denotes the probability distribution of Eve's observation sequence $\boldsymbol{Z}$ when Alice employs the innocent (resp. active) policy. The performance of Eve's detector can be characterized by the sum of its probability of missed detection and that of false alarm. To make the problem non-trivial, we assume that, roughly, the distribution of Eve's observation $Z_t$ when Alice sends an active probe $X_t=1$ can ``hide as noise" within that when Alice sends an innocent probe $X_t=0$. To make this precise, define Eve's observation channel with Alice's probing action $x$ and channel parameter $\theta$ as
\begin{equation}
    \label{eq:Eve's_channel}
    Q^x_\theta(z) \triangleq \sum_{y\in\mathcal{Y}} W_{Y,Z|X,\theta}(y,z|x,\theta),
    ~\forall \theta\in\Theta,~x\in\mathcal{X},~z\in\mathcal{Z}
    .
\end{equation}
Then this assumption translates to the absolute continuity conditions $Q^1_\theta \ll Q^0_\theta$, for $\theta\in\Theta$. If this were not the case, any active probes sent by Alice can be detected by Eve with certainty, trivializing the problem. Moreover, we assume that $Q_\theta^1(\cdot)\ne Q_\theta^0(\cdot)$ for $\theta\in\Theta$, otherwise Eve cannot detect Alice's active probing at all, for both pre-change $\theta=0$ and post-change $\theta=1$ cases. 
This assumption implies that $ \chi_{2,\theta} \triangleq \chi_2( Q^1_{\theta} \| Q^0_{\theta} ) >0$ for $\theta \in \Theta = \{0,1\}$.

We assume that Alice knows Eve's presence and is trying to employ her active policy to probe the channel parameter change while staying undetected (i.e., covert) w.r.t. Eve. To pose this as a non-trivial problem, our first major assumption is that of ``no free passive sensing", which, roughly speaking, is that the innocent action $0$ cannot provide any information about the channel change. To be precise, define Alice's observation channel with probing action $x$ and channel parameter $\theta$ as
\begin{equation}
    \label{eq:Alice's_channel}
    P^x_\theta(y) \triangleq \sum_{z\in\mathcal{Z}} W_{Y,Z|X,\theta}(y,z|x,\theta),
    ~\forall \theta\in\Theta,~x\in\mathcal{X},~y\in\mathcal{Y}
    .
\end{equation}
Then the ``no free passive sensing" assumption translates to the equality of PMFs $P^0_1(\cdot)=P^0_0(\cdot)$. If this were not the case, Alice could send the innocent symbol for all time and still obtain nonzero information on channel change via the change of observation statistics $P^0_1$ post-change versus $P^0_0$ pre-change, while maintaining covertness w.r.t. Eve trivially. We also assume the complementary condition that ``active sensing has non-zero but finite gain". To be precise, let $D\triangleq\mathbb{D}(P^1_1\|P^1_0)$, and this assumption means $0<D<\infty$. For technical reasons, we also assume that $V\triangleq\mathbb{V}(P^1_1\|P^1_0)<\infty$.

In this paper, we also assume that both Alice and Eve know the full channel statistics, i.e., they have access to the PMFs $W_{Y,Z|X,\theta}(\cdot,\cdot|x,\theta)$ for all $x\in\mathcal{X}$ and $\theta\in\Theta$. This means that, as a quickest change detection problem, both the pre-change and post-change distributions are known to both parties. We leave the relaxation of this assumption as future work.

To rigorously formalize the quickest change detection problem, we define the following sequence of probability measures. For any $k \in \mathbb{N}$, let $\mathbb{P}_k$ denote the probability measure conditioned on the change-point occurring at time $\Gamma = k$, and $\mathbb{E}_k$ the corresponding expectation. 
The overall Bayesian probability measure, which averages over the unknown change-point, is denoted by $\mathbb{P}_\pi$ (or simply $\mathbb{P}$), and let $\mathbb{E}_\pi$ (or simply $\mathbb{E}$) denote the corresponding expectation. Under $\mathbb{P}_\pi$, the change-point $\Gamma$ follows the geometric prior distribution $\pi$, such that for any event $\mathcal{E}$, $\mathbb{P}_\pi(\mathcal{E}) = \sum_{k\in\mathbb{N}} \pi_k \mathbb{P}_k(\mathcal{E})$. To stress that the probability (resp. expectation) is induced by Alice's policy $\Pi$, we also write $\mathbb{P}^{\Pi}$/$\mathbb{P}_k^{\Pi}$ (resp. $\mathbb{E}^{\Pi}$/$\mathbb{E}_k^{\Pi}$).

We call the problem of Alice designing her active policy to probe the channel change in a quickest manner, while maintaining covertness w.r.t. Eve, the covert Bayesian quickest change detection (covert BQCD) problem. Before formalizing this design as an optimization problem, we investigate the covertness constraint imposed by Eve's detection in the following section.

\section{Covertness Constraint Analysis}
\label{sec:covert_constr_analys}

The following analysis is inspired by a similar analysis in the finite-horizon setting~\cite{Tahmasbi2020b}.

Recall Eve's binary hypothesis test~\eqref{eq:eve_hypo_test_inf_length}: 
\begin{align}
    \begin{split}
    &H_0:~\boldsymbol{Z}\sim \mathbb{P}^{\Pi^{\boldsymbol{0}}}_{\boldsymbol{Z}},\\
    &H_1:~\boldsymbol{Z}\sim \mathbb{P}^{\Pi^{\textnormal{A}}}_{\boldsymbol{Z}}
    .
    \end{split}
\end{align}
A metric to quantify the performance of Eve's binary test is the total probability of error (see, e.g.,~\cite[eq.~(7.17)]{polyanskiy2025information}), i.e., the sum of probability of false alarm $P_\textnormal{Eve, False Alarm}$ and probability of missed detection $P_\textnormal{Eve, Missed Detection}$. Both probabilities are functions of the rejection region $\mathcal{R}$, which is a subset of the sequence space $\mathcal{Z}^\mathbb{N}$ where Eve rejects $H_0$ whenever she observes a sequence in $\mathcal{R}$. Thus,
\begin{align}
    P_\textnormal{Eve, Missed Detection} 
    &\triangleq  \mathbb{P}^{\Pi^\textnormal{A}}_{\boldsymbol{Z}}(\mathcal{Z}^\mathbb{N}\setminus\mathcal{R}),\\
    P_\textnormal{Eve, False Alarm} 
    &\triangleq  \mathbb{P}^{\Pi^{\boldsymbol{0}}}_{\boldsymbol{Z}}(\mathcal{R})
    .
\end{align}
Eve designs her detector $\mathscr{D}$ to minimize the total probability of error:
\begin{align}
    \inf_{\textnormal{Eve's test}} ( P_\textnormal{Eve, Missed Detection} + P_\textnormal{Eve, False Alarm} )
    &=
    \inf_{\mathcal{R}\subset\mathcal{Z}^\mathbb{N}}
    \left(
    \mathbb{P}^{\Pi^\textnormal{A}}_{\boldsymbol{Z}}(\mathcal{Z}^\mathbb{N}\setminus\mathcal{R})
    +
    \mathbb{P}^{\Pi^{\boldsymbol{0}}}_{\boldsymbol{Z}}(\mathcal{R})
    \right)
    \\&
    =
    1 - \textnormal{TV}
    \left( \mathbb{P}^{\Pi^\textnormal{A}}_{\boldsymbol{Z}}
    ,
    \mathbb{P}^{\Pi^{\boldsymbol{0}}}_{\boldsymbol{Z}}
    \right)
    \\&\geq
    1
    -
    \sqrt{\frac{1}{2}
    \mathbb{D}
    \left( \mathbb{P}^{\Pi^\textnormal{A}}_{\boldsymbol{Z}}
    \Big\|
    \mathbb{P}^{\Pi^{\boldsymbol{0}}}_{\boldsymbol{Z}}
    \right)
    }
    ,
\end{align}
where the penultimate equality follows from the variational characterization of the total variation distance (see, e.g.,~\cite[Theorem~7.7]{polyanskiy2025information}), and the last inequality is Pinsker's. Therefore, a reasonable covertness constraint may be set as
\begin{align}
    \mathbb{D}
    \left( \mathbb{P}^{\Pi^\textnormal{A}}_{\boldsymbol{Z}}
    \Big\|
    \mathbb{P}^{\Pi^{\boldsymbol{0}}}_{\boldsymbol{Z}}
    \right)
    \leq
    \delta
\end{align}
for some $\delta>0$. Also see~\cite[Appendix A]{Bloch2015b} for a discussion on the relation between relative entropy and hypothesis testing.

Directly constraining the relative entropy
$\mathbb{D}\left(\mathbb{P}^{\Pi^\textnormal{A}}_{\boldsymbol{Z}}\big\|\mathbb{P}^{\Pi^{\boldsymbol{0}}}_{\boldsymbol{Z}}\right)$
is analytically difficult because a closed-loop probing policy generally induces statistical dependence across the probing sequence $(X_t)_{t\in\mathbb{N}}$. To obtain a tractable surrogate constraint, we first analyze
$\mathbb{D}\left(\mathbb{P}^{\Pi^\textnormal{A}}_{Z^N}\big\|\mathbb{P}^{\Pi^{\boldsymbol{0}}}_{Z^N}\right)$,
the relative entropy between distributions on the truncated sequence $Z^N$ for some observation window of length $N\in\mathbb{N}$, and then upper bound it by a quantity depending only on the instantaneous probing probabilities. Then, by the monotone convergence theorem for relative entropy~\cite[Theorem~4.11]{polyanskiy2025information}, i.e.,
\begin{align}
    \lim_{N\to\infty}
    \mathbb{D}
    \left( \mathbb{P}^{\Pi^\textnormal{A}}_{Z^N}
    \Big\|
    \mathbb{P}^{\Pi^{\boldsymbol{0}}}_{Z^N}
    \right)
    =
    \mathbb{D}
    \left( \mathbb{P}^{\Pi^\textnormal{A}}_{\boldsymbol{Z}}
    \Big\|
    \mathbb{P}^{\Pi^{\boldsymbol{0}}}_{\boldsymbol{Z}}
    \right)
    ,
\end{align}
we can take the observation window length $N$ to infinity to obtain a corresponding upper bound on $\mathbb{D}\left(\mathbb{P}^{\Pi^\textnormal{A}}_{\boldsymbol{Z}}\big\|\mathbb{P}^{\Pi^{\boldsymbol{0}}}_{\boldsymbol{Z}}\right)$.

Next, for any fixed $N\in\mathbb{N}$, we derive an analytically tractable upper bound for $\mathbb{D}\left(\mathbb{P}^{\Pi^\textnormal{A}}_{Z^N}\Big\|\mathbb{P}^{\Pi^{\boldsymbol{0}}}_{Z^N}\right)$. First, observe that by conditioning on the changepoint,
\begin{align}
    \mathbb{P}^{\Pi^\textnormal{A}}_{Z^N}
    &=
    \sum_{k\in\mathbb{N}} \pi_k \mathbb{P}^{\Pi^\textnormal{A}}_{Z^N|\Gamma=k}
    ;\\
    \mathbb{P}^{\Pi^{\boldsymbol{0}}}_{Z^N}
    &=
    \sum_{k\in\mathbb{N}} \pi_k \mathbb{P}^{\Pi^{\boldsymbol{0}}}_{Z^N|\Gamma=k}
    ,
\end{align}
and thus
\begin{align}
    \label{eq:joint_convex_inf}
    \mathbb{D}
    \left(
    \mathbb{P}^{\Pi^\textnormal{A}}_{Z^N}
    \Big\|
    \mathbb{P}^{\Pi^{\boldsymbol{0}}}_{Z^N}
    \right)
\leq
    \sum_{k\in\mathbb{N}}
        \pi_k
        \mathbb{D}
        \left(
        \mathbb{P}^{\Pi^\textnormal{A}}_{Z^N|\Gamma=k}
        \Big\|
        \mathbb{P}^{\Pi^{\boldsymbol{0}}}_{Z^N|\Gamma=k}
        \right)
\end{align}
by applying the joint convexity of relative entropy~\cite[Theorem~7.5]{polyanskiy2025information}. Since both distributions are infinite mixtures, one can first truncate the mixture up to term $M$ and then renormalize to get valid PMFs, apply the joint convexity of relative entropy between finite mixtures~\cite[Theorem~7.5]{polyanskiy2025information}, then take $M\to\infty$ to obtain~\eqref{eq:joint_convex_inf}. As this is standard, we omit the proof.

We thus turn our attention to upper bound $\mathbb{D}
        \left(
        \mathbb{P}^{\Pi^\textnormal{A}}_{Z^N|\Gamma=k}
        \Big\|
        \mathbb{P}^{\Pi^{\boldsymbol{0}}}_{Z^N|\Gamma=k}\right)$
for any fixed $k\in\mathbb{N}$.
By the chain rule for relative entropy~\cite[Theorem~2.16]{polyanskiy2025information},
\begin{align}
&\label{eq:covert_analysis_conditional_relative_entropy}
\mathbb{D}
    \left(
    \mathbb{P}^{\Pi^\textnormal{A}}_{Z^N|\Gamma=k}
    \Big\|
    \mathbb{P}^{\Pi^{\boldsymbol{0}}}_{Z^N|\Gamma=k}
    \right)
=
\sum_{i=1}^{N}
\mathbb{D}
\left(
    \mathbb{P}^{\Pi^\textnormal{A}}_{Z_{i}|Z^{i-1},\Gamma=k}
\Big\|
    \mathbb{P}^{\Pi^{\boldsymbol{0}}}_{Z_{i}|Z^{i-1},\Gamma=k}
\Big|
    \mathbb{P}^{\Pi^\textnormal{A}}_{Z^{i-1}|\Gamma=k}
\right)
\\
=~&
\sum_{i=1}^{N}
\sum_{z^{i-1}\in\mathcal{Z}^{i-1}}
\mathbb{P}^{\Pi^\textnormal{A}}_{Z^{i-1}|\Gamma=k}(z^{i-1})
\mathbb{D}
\left(
    \mathbb{P}^{\Pi^\textnormal{A}}_{Z_{i}|Z^{i-1}=z^{i-1},\Gamma=k}
\Big\|
    \mathbb{P}^{\Pi^{\boldsymbol{0}}}_{Z_{i}|Z^{i-1}=z^{i-1},\Gamma=k}
\right)
,
\end{align}
where for each $z^{i-1}\in\mathcal{Z}^{i-1}$, under $\Pi^{\boldsymbol{0}}$ and conditioned on $\Gamma=k$, the channel outputs are independent with time-varying marginal
$Q^0_{\mathbb{I}\{i>k\}}$ at time $i$. Hence,
\begin{align}
&
\mathbb{D}
\left(
    \mathbb{P}^{\Pi^\textnormal{A}}_{Z_{i}|Z^{i-1}=z^{i-1},\Gamma=k}
\Big\|
    \mathbb{P}^{\Pi^{\boldsymbol{0}}}_{Z_{i}|Z^{i-1}=z^{i-1},\Gamma=k}
\right)
=
\mathbb{D}
\left(
    \mathbb{P}^{\Pi^\textnormal{A}}_{Z_{i}|Z^{i-1}=z^{i-1},\Gamma=k}
\Big\|
    \mathbb{P}^{\Pi^{\boldsymbol{0}}}_{Z_{i}|\Gamma=k}
\right)
\\
\triangleq~&\label{eq:covert_analysis_relative_entropy}
\sum_{z_i\in\mathcal{Z}}
\mathbb{P}^{\Pi^\textnormal{A}}_{Z_{i}|Z^{i-1}=z^{i-1},\Gamma=k}(z_i)
\ln
\left(
\frac{
\mathbb{P}^{\Pi^\textnormal{A}}_{Z_{i}|Z^{i-1}=z^{i-1},\Gamma=k}(z_i)
}{\mathbb{P}^{\Pi^{\boldsymbol{0}}}_{Z_{i}|\Gamma=k}(z_i)}
\right)
.
\end{align}
So we focus on the conditional PMF $\mathbb{P}^{\Pi^\textnormal{A}}_{Z_{i}|Z^{i-1}=z^{i-1},\Gamma=k}(z_i)\triangleq\mathbb{P}^{\Pi^\textnormal{A}}_{Z_i|Z^{i-1},\Gamma=k}(z_i|z^{i-1})$.
By the law of total probability, we can further express
\begin{align}
&
    \mathbb{P}^{\Pi^\textnormal{A}}_{Z_i|Z^{i-1},\Gamma=k}(z_i|z^{i-1})
=
    \sum_{x^i\in\mathcal{X}^i}
    \sum_{y^{i-1}\in\mathcal{Y}^{i-1}}
    \mathbb{P}^{\Pi^\textnormal{A}}_{X^i,Y^{i-1},Z_i|Z^{i-1},\Gamma=k}(x^i,y^{i-1},z_i|z^{i-1}) 
\\=~&\label{eq:covert_analysis_conditional_PMF_expansion}
    \sum_{x^i\in\mathcal{X}^i}
    \sum_{y^{i-1}\in\mathcal{Y}^{i-1}}
    \mathbb{P}^{\Pi^\textnormal{A}}_{X^i,Y^{i-1}|Z^{i-1},\Gamma=k}(x^i,y^{i-1}|z^{i-1}) 
    \cdot
    \mathbb{P}^{\Pi^\textnormal{A}}_{Z_i|X^i,Y^{i-1},Z^{i-1},\Gamma=k}(z_i|x^i,y^{i-1},z^{i-1})
.
\end{align}
Note that, by the discrete memoryless channel (DMC) assumption, we have
\begin{align}
    \label{eq:covert_analysis_DMC}
    \mathbb{P}^{\Pi^\textnormal{A}}_{Z_i|X^i,Y^{i-1},Z^{i-1},\Gamma=k}(z_i|x^i,y^{i-1},z^{i-1})
=
    Q^{x_i}_{ \theta_i(k) }(z_i)
=
    Q^{x_i}_{ \mathbb{I}\{i> k\} }(z_i)
    .
\end{align}
Also, by the definition of conditional PMF,
\begin{align}
    \label{eq:covert_analysis_conditional_PMF}
    \begin{split}
&
    \mathbb{P}^{\Pi^\textnormal{A}}_{X^i,Y^{i-1}|Z^{i-1},\Gamma=k}(x^i,y^{i-1}|z^{i-1}) 
\\
=~&
    \mathbb{P}^{\Pi^\textnormal{A}}_{X^{i-1},Y^{i-1}|Z^{i-1},\Gamma=k}(x^{i-1},y^{i-1}|z^{i-1})
    \cdot
    \mathbb{P}^{\Pi^\textnormal{A}}_{X_i|X^{i-1},Y^{i-1},Z^{i-1},\Gamma=k}(x_i|x^{i-1},y^{i-1},z^{i-1}) 
.
    \end{split}
\end{align}
Since the probing action $X_i\sim\textnormal{Ber}(\beta_i(X^{i-1},Y^{i-1}))$ is conditionally independent of $\boldsymbol{Z}$ and $\Gamma$ given $(X^{i-1},Y^{i-1})$, we have
\begin{align}
    \label{eq:covert_analysis_Alice_active_policy_1}
    \mathbb{P}^{\Pi^\textnormal{A}}_{X_i|X^{i-1},Y^{i-1},Z^{i-1},\Gamma=k}(x_i|x^{i-1},y^{i-1},z^{i-1})
=
    \mathbb{P}^{\Pi^\textnormal{A}}_{X_i|X^{i-1},Y^{i-1}}(x_i|x^{i-1},y^{i-1})
    .
\end{align}
For the target conditional PMF~\eqref{eq:covert_analysis_conditional_PMF_expansion} expression,
substituting the conditional PMF expressions~\eqref{eq:covert_analysis_conditional_PMF}, ~\eqref{eq:covert_analysis_Alice_active_policy_1} and the simplification via DMC~\eqref{eq:covert_analysis_DMC} yields
\begin{align}
&
    \mathbb{P}^{\Pi^\textnormal{A}}_{Z_i|Z^{i-1},\Gamma=k}(z_i|z^{i-1})
=
    \sum_{x^{i-1}\in\mathcal{X}^{i-1}}
    \sum_{y^{i-1}\in\mathcal{Y}^{i-1}}
    \mathbb{P}^{\Pi^\textnormal{A}}_{X^{i-1},Y^{i-1}|Z^{i-1},\Gamma=k}(x^{i-1},y^{i-1}|z^{i-1})
    \nonumber\\&\qquad\qquad\qquad\qquad\qquad\times
    \sum_{x_i\in\mathcal{X}}
    \mathbb{P}^{\Pi^\textnormal{A}}_{X_i|X^{i-1},Y^{i-1}}(x_i|x^{i-1},y^{i-1})
    \cdot
    Q^{x_i}_{ \mathbb{I}\{i> k\} }(z_i)
\\=~&\label{eq:covert_analysis_conditional_PMF_Zi_Zi-1}
    \mathbb{E}^{\Pi^\textnormal{A}}_{X^{i-1},Y^{i-1}|Z^{i-1}=z^{i-1},\Gamma=k}
    \left[
        \sum_{x_i\in\mathcal{X}}
        \mathbb{P}^{\Pi^\textnormal{A}}_{X_i|X^{i-1},Y^{i-1}}(x_i|X^{i-1},Y^{i-1})
        \cdot
        Q^{x_i}_{ \mathbb{I}\{i> k\} }(z_i)
    \right]
.
\end{align}
Expressing Alice's active policy with indicator functions on her decision variables $d_i(x^{i-1},y^{i-1})$ and sensing probability $\beta_i(x^{i-1},y^{i-1})$, we have
\begin{align}
    \label{eq:covert_analysis_active_policy}
    \begin{split}
    \mathbb{P}^{\Pi^\textnormal{A}}_{X_i|X^{i-1},Y^{i-1}}(x_i|x^{i-1},y^{i-1})
&=
    \left(
        1
    -
        \mathbb{I}
        \left\{ d_i(x^{i-1},y^{i-1}) = 0 \right\}
        \beta_i(x^{i-1},y^{i-1})
    \right)
    \cdot \mathbb{I}\{x_i=0\}
\\&~+
    \mathbb{I}
    \left\{ d_i(x^{i-1},y^{i-1}) = 0 \right\}
    \beta_i(x^{i-1},y^{i-1})\cdot \mathbb{I}\{x_i=1\}
    ,
    \end{split}
\end{align}
since Alice only actively probes ($x_i=1$) if she decides to continue and the $\textnormal{Ber}(\beta_i)$ realization is $1$. From this, we can further compute the formula inside the expectation~\eqref{eq:covert_analysis_conditional_PMF_Zi_Zi-1} as
\begin{align}
&
    \sum_{x_i\in\mathcal{X}}
        \mathbb{P}^{\Pi^\textnormal{A}}_{X_i|X^{i-1},Y^{i-1}}(x_i|X^{i-1},Y^{i-1})
        \cdot
        Q^{x_i}_{ \mathbb{I}\{i> k\} }(z_i)
\nonumber\\=~&
    \left(
        1
    -
        \mathbb{I}
        \left\{ d_i(X^{i-1},Y^{i-1}) = 0 \right\}
        \beta_i(X^{i-1},Y^{i-1})
    \right)
    \cdot 
    Q^{0}_{ \mathbb{I}\{i> k\} }(z_i)
\nonumber\\&+
    \mathbb{I}
    \left\{ d_i(X^{i-1},Y^{i-1}) = 0 \right\}
    \beta_i(X^{i-1},Y^{i-1})
    \cdot 
    Q^{1}_{ \mathbb{I}\{i> k\} }(z_i)
\\=~&\label{eq:covert_analysis_bernoulli_decomp}
    Q^{0}_{ \mathbb{I}\{i> k\} }(z_i)
+
    \mathbb{I}
    \left\{ d_i(X^{i-1},Y^{i-1}) = 0 \right\}
    \beta_i(X^{i-1},Y^{i-1})
    \left(
        Q^{1}_{ \mathbb{I}\{i> k\} }(z_i)
    -
        Q^{0}_{ \mathbb{I}\{i> k\} }(z_i)
    \right)
.
\end{align}
Applying~\eqref{eq:covert_analysis_bernoulli_decomp} in~\eqref{eq:covert_analysis_conditional_PMF_Zi_Zi-1}, by linearity of expectation, we have
\begin{align}
&
    \mathbb{P}^{\Pi^\textnormal{A}}_{Z_i|Z^{i-1},\Gamma=k}(z_i|z^{i-1})
=
    Q^{0}_{ \mathbb{I}\{i> k\} }(z_i)
    + 
    \left( 
        Q^{1}_{ \mathbb{I}\{i> k\} }(z_i)
        -
        Q^{0}_{ \mathbb{I}\{i> k\} }(z_i)
    \right)
\nonumber\\&\qquad\qquad\qquad\qquad\quad\times
    \mathbb{E}^{\Pi^\textnormal{A}}_{X^{i-1},Y^{i-1}|Z^{i-1}=z^{i-1},\Gamma=k}
    \left[
        \mathbb{I}
        \left\{ d_i(X^{i-1},Y^{i-1}) = 0 \right\}
        \beta_i(X^{i-1},Y^{i-1})
    \right]
.
\end{align}
For brevity, we introduce the following definition.
\begin{definition}
    For $i\in\mathbb{N}$, define the deterministic function $\tilde{\beta}_{i}:\mathcal{Z}^{i-1}\times\mathbb{N}\to[0,1]$ by
    \begin{align}
        \tilde{\beta}_{i}( z^{i-1}, k)
    \triangleq
        \mathbb{E}^{\Pi^\textnormal{A}}_{X^{i-1},Y^{i-1}|Z^{i-1},\Gamma}
        \left[
            \mathbb{I}
            \left\{ d_i(X^{i-1},Y^{i-1}) = 0 \right\}
            \beta_i(X^{i-1},Y^{i-1})
        \Big\vert
            z^{i-1},k
        \right]
        .
    \end{align}
\end{definition}
\begin{remark}
   We note that $\tilde{\beta}_i$ can be understood as Eve's estimate of the probability that Alice is active probing at time $i$ given Eve's causal observation $z^{i-1}$, conditioned on $\Gamma=k$. Also note that $(\tilde{\beta}_i)_{i\in\mathbb{N}}$ is parameterized by the policy $\Pi^\textnormal{A}$.
\end{remark}

With this definition, the conditional PMF of $Z_i$ given $Z^{i-1}=z^{i-1}$ and $\Gamma=k$ under policy $\Pi^\textnormal{A}$ can be expressed as
\begin{align}
    \label{eq:covert_analysis_conditional_PMF_Zi_Zi-1_2}
    \mathbb{P}^{\Pi^\textnormal{A}}_{Z_i|Z^{i-1},\Gamma=k}(z_i|z^{i-1})
=
    Q^{0}_{ \mathbb{I}\{i> k\} }(z_i)
    + 
    \tilde{\beta}_{i}( z^{i-1}, k)
    \left( 
        Q^{1}_{ \mathbb{I}\{i> k\} }(z_i)
        -
        Q^{0}_{ \mathbb{I}\{i> k\} }(z_i)
    \right)
    .
\end{align}
Noting that under $\Pi^{\boldsymbol{0}}$ we have
$\mathbb{P}^{\Pi^{\boldsymbol{0}}}_{Z_{i}|\Gamma=k}(z_i) = Q^0_{\mathbb{I}\{i> k\}}(z_i)$, and hence, the log-likelihood ratio in~\eqref{eq:covert_analysis_relative_entropy} is
\begin{align}
    \label{eq:covert_analysis_log_like}
&
    \ln
        \left(
            \frac
                { \mathbb{P}^{\Pi^\textnormal{A}}_{Z_{i}|Z^{i-1}=z^{i-1},\Gamma=k}(z_i) }
                { \mathbb{P}^{\Pi^{\boldsymbol{0}}}_{Z_{i}|\Gamma=k}(z_i) }
        \right)
=
    \ln
        \left(
            1
        + 
            \tilde{\beta}_i( z^{i-1}, k)
            \frac
                {   Q^{1}_{ \mathbb{I}\{i> k\} }(z_i) 
                -
                    Q^{0}_{ \mathbb{I}\{i> k\} }(z_i)
                }
                { Q^{0}_{ \mathbb{I}\{i> k\} }(z_i) }
        \right)
.
\end{align}
Therefore, substituting~\eqref{eq:covert_analysis_conditional_PMF_Zi_Zi-1_2} and~\eqref{eq:covert_analysis_log_like} into~\eqref{eq:covert_analysis_relative_entropy} yields
\begin{align}
    \mathbb{D}
    \left(
        \mathbb{P}^{\Pi^\textnormal{A}}_{Z_{i}|Z^{i-1}=z^{i-1},\Gamma=k}
    \Big\|
        \mathbb{P}^{\Pi^{\boldsymbol{0}}}_{Z_{i}|\Gamma=k}
    \right)
=&
    \sum_{z_i\in\mathcal{Z}}
        \left(
            Q^{0}_{ \mathbb{I}\{i> k\} }(z_i)
        +
            \tilde{\beta}_i(z^{i-1},k)
            \left( 
                Q^{1}_{ \mathbb{I}\{i> k\} }(z_i)
                -
                Q^{0}_{ \mathbb{I}\{i> k\} }(z_i)
            \right)
        \right)
\nonumber\\&\qquad\times
    \ln
        \left(
            1
        +
            \tilde{\beta}_i(z^{i-1},k)
            \frac
            { 
                Q^{1}_{ \mathbb{I}\{i> k\} }(z_i)
                -
                Q^{0}_{ \mathbb{I}\{i> k\} }(z_i)
            }{ 
                Q^{0}_{ \mathbb{I}\{i> k\} }(z_i)
            }
        \right)
.
\end{align}
To get a tractable upper bound, we follow the same idea as in~\cite[eq.~(118)-(120)]{Bloch2015b} but with the simpler inequality $\ln(1+x)\leq x$ for $x>-1$,
\begin{align}
&
    \mathbb{D}
    \left(
        \mathbb{P}^{\Pi^\textnormal{A}}_{Z_{i}|Z^{i-1}=z^{i-1},\Gamma=k}
    \Big\|
        \mathbb{P}^{\Pi^{\boldsymbol{0}}}_{Z_{i}|\Gamma=k}
    \right)
\leq
    \left(
        \tilde{\beta}_i(z^{i-1},k)
    \right)^2 
    \chi_2 
        \left( 
            Q^{1}_{ \mathbb{I}\{i> k\} } 
        \Big\|
            Q^{0}_{ \mathbb{I}\{i> k\} }
        \right)
.
\end{align}
The boundary case in which the logarithm argument is zero is handled by continuity, with the convention $0\ln0\triangleq0$.
Therefore, the conditional relative entropy in~\eqref{eq:covert_analysis_conditional_relative_entropy} can be upper bounded as
\begin{align}
&
    \mathbb{D}
    \left(
        \mathbb{P}^{\Pi^\textnormal{A}}_{Z_{i}|Z^{i-1},\Gamma=k}
    \Big\|
        \mathbb{P}^{\Pi^{\boldsymbol{0}}}_{Z_{i}|Z^{i-1},\Gamma=k}
    \Big|
        \mathbb{P}^{\Pi^\textnormal{A}}_{Z^{i-1}|\Gamma=k}
    \right)
\nonumber\\\triangleq~&
    \sum_{z^{i-1}\in\mathcal{Z}^{i-1}}
    \mathbb{P}^{\Pi^\textnormal{A}}_{Z^{i-1}|\Gamma=k}(z^{i-1})
    \mathbb{D}
    \left(
        \mathbb{P}^{\Pi^\textnormal{A}}_{Z_{i}|Z^{i-1}=z^{i-1},\Gamma=k}
    \Big\|
        \mathbb{P}^{\Pi^{\boldsymbol{0}}}_{Z_{i}|Z^{i-1}=z^{i-1},\Gamma=k}
    \right)
\\\leq~& 
    \chi_2 
        \left( 
            Q^{1}_{ \mathbb{I}\{i> k\} } 
        \Big\|
            Q^{0}_{ \mathbb{I}\{i> k\} }
        \right)
    \sum_{z^{i-1}\in\mathcal{Z}^{i-1}}
    \mathbb{P}^{\Pi^\textnormal{A}}_{Z^{i-1}|\Gamma=k}(z^{i-1})
    \left(
        \tilde{\beta}_i(z^{i-1},k)
    \right)^2
\\=~&
    \chi_2 
        \left( 
            Q^{1}_{ \mathbb{I}\{i> k\} } 
        \Big\|
            Q^{0}_{ \mathbb{I}\{i> k\} }
        \right)
    \mathbb{E}^{\Pi^\textnormal{A}}_{Z^{i-1}|\Gamma=k}
    \left[
        \left(
            \tilde{\beta}_i(Z^{i-1},k)
        \right)^2
    \right]
.
\end{align}
Applying this bound in~\eqref{eq:covert_analysis_conditional_relative_entropy} yields
\begin{align}
    \label{eq:covert_analysis_cond_rela_entr_upp_bdd_1}
    \mathbb{D}
    \left(
    \mathbb{P}^{\Pi^\textnormal{A}}_{Z^N|\Gamma=k}
    \Big\|
    \mathbb{P}^{\Pi^{\boldsymbol{0}}}_{Z^N|\Gamma=k}
    \right)
\le
    \sum_{i=1}^{N}
    \chi_2 
        \left( 
            Q^{1}_{ \mathbb{I}\{i> k\} } 
        \Big\|
            Q^{0}_{ \mathbb{I}\{i> k\} }
        \right)
    \mathbb{E}^{\Pi^\textnormal{A}}_{Z^{i-1}|\Gamma=k}
    \left[
        \left(
            \tilde{\beta}_i(Z^{i-1},k)
        \right)^2
    \right]
.
\end{align}
For any $z^{i-1}\in\mathcal{Z}^{i-1}$, by convexity of $x\mapsto x^2$,
\begin{align}
    \left(
        \tilde{\beta}_i(z^{i-1},k)
        \right)^2
&\triangleq   
    \left(
        \mathbb{E}^{\Pi^\textnormal{A}}_{X^{i-1},Y^{i-1}|Z^{i-1},\Gamma=k}
        \left[
            \mathbb{I}
            \left\{ d_i(X^{i-1},Y^{i-1}) = 0 \right\}
            \beta_i(X^{i-1},Y^{i-1})
        \Big\vert
            z^{i-1}
        \right]
    \right)^2
\\&\leq
    \mathbb{E}^{\Pi^\textnormal{A}}_{X^{i-1},Y^{i-1}|Z^{i-1},\Gamma=k}
        \left[
            \left(
            \mathbb{I}
            \left\{ d_i(X^{i-1},Y^{i-1}) = 0 \right\}
            \beta_i(X^{i-1},Y^{i-1})
            \right)^2
        \Big\vert
            z^{i-1}
        \right]
\\&=\label{eq:covert_analysis_tild_beta_squa_upp_bdd}
    \mathbb{E}^{\Pi^\textnormal{A}}_{X^{i-1},Y^{i-1}|Z^{i-1},\Gamma=k}
        \left[
            \mathbb{I}
            \left\{ d_i(X^{i-1},Y^{i-1}) = 0 \right\}
            \left(
            \beta_i(X^{i-1},Y^{i-1})
            \right)^2
        \Big\vert
            z^{i-1}
        \right]
    .
\end{align}
Using~\eqref{eq:covert_analysis_tild_beta_squa_upp_bdd} in~\eqref{eq:covert_analysis_cond_rela_entr_upp_bdd_1} yields
\begin{align}
&
    \mathbb{D}
    \left(
    \mathbb{P}^{\Pi^\textnormal{A}}_{Z^N|\Gamma=k}
    \Big\|
    \mathbb{P}^{\Pi^{\boldsymbol{0}}}_{Z^N|\Gamma=k}
    \right)
\nonumber\\\leq~&
    \sum_{i=1}^{N}
    \chi_2 
        \left( 
            Q^{1}_{ \mathbb{I}\{i> k\} } 
        \Big\|
            Q^{0}_{ \mathbb{I}\{i> k\} }
        \right)
\nonumber\\\times~&
    \mathbb{E}^{\Pi^\textnormal{A}}_{Z^{i-1}|\Gamma=k}
    \left[
         \mathbb{E}^{\Pi^\textnormal{A}}_{X^{i-1},Y^{i-1}|Z^{i-1},\Gamma=k}
        \left[
            \mathbb{I}
            \left\{ d_i(X^{i-1},Y^{i-1}) = 0 \right\}
            \left(
            \beta_i(X^{i-1},Y^{i-1})
            \right)^2
        \Big\vert
            Z^{i-1}
        \right]
    \right]
\\\overset{\textnormal{(a)}}{=}~&
    \sum_{i=1}^{N}
    \chi_2 
        \left( 
            Q^{1}_{ \mathbb{I}\{i> k\} } 
        \Big\|
            Q^{0}_{ \mathbb{I}\{i> k\} }
        \right)
    \mathbb{E}^{\Pi^\textnormal{A}}_{k}
    \left[
        \mathbb{I}
            \left\{ d_i(X^{i-1},Y^{i-1}) = 0 \right\}
        \left(
            \beta_i(X^{i-1},Y^{i-1})
        \right)^2
    \right]
\\\overset{\textnormal{(b)}}{=}~&
    \sum_{i=1}^{N}
    \chi_2 
        \left( 
            Q^{1}_{ \mathbb{I}\{i> k\} } 
        \Big\|
            Q^{0}_{ \mathbb{I}\{i> k\} }
        \right)
    \mathbb{E}^{\Pi^\textnormal{A}}_{k}
    \left[
        \mathbb{I}
            \left\{ \tau^\textnormal{A} > i \right\}
        \left(
            \beta_i(X^{i-1},Y^{i-1})
        \right)^2
    \right]
\\\overset{\textnormal{(c)}}{=}~&
    \mathbb{E}^{\Pi^\textnormal{A}}_{k}
    \left[
    \sum_{i=1}^{N}
    \chi_2 
        \left( 
            Q^{1}_{ \mathbb{I}\{i> k\} } 
        \Big\|
            Q^{0}_{ \mathbb{I}\{i> k\} }
        \right)
        \mathbb{I}
            \left\{ \tau^\textnormal{A} > i \right\}
        \left(
            \beta_i(X^{i-1},Y^{i-1})
        \right)^2
    \right]
\\\overset{\textnormal{(d)}}{=}~&\label{eq:covert_analysis_rela_entr_upp_bdd_cond_changepoint}
    \mathbb{E}^{\Pi^\textnormal{A}}_{k}
    \left[
    \sum_{i=1}^{\min\{\tau^\textnormal{A},N\}}
    \chi_2 
        \left( 
            Q^{1}_{ \mathbb{I}\{i> k\} } 
        \Big\|
            Q^{0}_{ \mathbb{I}\{i> k\} }
        \right)
        \left(
            \beta_i(X^{i-1},Y^{i-1})
        \right)^2
    \right]
,
\end{align}
where (a) follows from the tower property, (b) follows from the relation between Alice's stopping time $\tau^\textnormal{A}$ and the stopping variables $d_i$, (c) follows from linearity of expectation, and (d) follows from our convention that Alice sends the innocent symbol after stopping, so the effective sensing probability vanishes beyond the stopping epoch.

We note that this bound~\eqref{eq:covert_analysis_rela_entr_upp_bdd_cond_changepoint} could also be useful if we wish to consider an alternative covertness metric that is robust to the changepoint realization $k$, such as considering bounding its supremum over $k\in\mathbb{N}$ by $\delta$ instead.

Finally, applying the bound~\eqref{eq:covert_analysis_rela_entr_upp_bdd_cond_changepoint} in~\eqref{eq:joint_convex_inf} yields 
\begin{align}
    \mathbb{D}
    \left(
    \mathbb{P}^{\Pi^\textnormal{A}}_{Z^N}
    \Big\|
    \mathbb{P}^{\Pi^{\boldsymbol{0}}}_{Z^N}
    \right)
&\le
    \sum_{k\in\mathbb{N}}
        \pi_k
           \mathbb{E}^{\Pi^\textnormal{A}}_{k}
    \left[
    \sum_{i=1}^{\min\{\tau^\textnormal{A},N\}}
    \chi_2 
        \left( 
            Q^{1}_{ \mathbb{I}\{i> k\} } 
        \Big\|
            Q^{0}_{ \mathbb{I}\{i> k\} }
        \right)
        \left(
            \beta_i(X^{i-1},Y^{i-1})
        \right)^2
    \right]    
\\&~=
    \mathbb{E}^{\Pi^\textnormal{A}}
    \left[
    \sum_{i=1}^{\min\{\tau^\textnormal{A},N\}}
    \chi_2 
        \left( 
            Q^{1}_{ \mathbb{I}\{i> \Gamma\} } 
        \Big\|
            Q^{0}_{ \mathbb{I}\{i> \Gamma\} }
        \right)
        \left(
            \beta_i(X^{i-1},Y^{i-1})
        \right)^2
    \right]
    ,
\end{align}
which, after taking $N\to\infty$ and applying monotone convergence theorem on the right-hand side (we assume that $\tau^\textnormal{A}<\infty$ with probability $1$, otherwise Alice would have an infinite average detection delay, which is inadmissible; see Section~\ref{sec:opt_prob}), we obtain
\begin{align}
    \mathbb{D}
    \left(
    \mathbb{P}^{\Pi^\textnormal{A}}_{\boldsymbol{Z}}
    \Big\|
    \mathbb{P}^{\Pi^{\boldsymbol{0}}}_{\boldsymbol{Z}}
    \right)
&\le
    \mathbb{E}^{\Pi^\textnormal{A}}
    \left[
    \sum_{i=1}^{\tau^\textnormal{A}}
    \chi_2 
        \left( 
            Q^{1}_{ \mathbb{I}\{i> \Gamma\} } 
        \Big\|
            Q^{0}_{ \mathbb{I}\{i> \Gamma\} }
        \right)
        \left(
            \beta_i(X^{i-1},Y^{i-1})
        \right)^2
    \right]
    .
\end{align}
Therefore, one way to ensure the covertness constraint $ 
    \mathbb{D}
    \left(
    \mathbb{P}^{\Pi^\textnormal{A}}_{\boldsymbol{Z}}
    \Big\|
    \mathbb{P}^{\Pi^{\boldsymbol{0}}}_{\boldsymbol{Z}}
    \right)\leq\delta$ 
is to optimize Alice's active sensing policy $\Pi^\textnormal{A}$ using some quickest change detection metrics while enforcing the following constraint
\begin{align}
    \mathbb{E}^{\Pi^\textnormal{A}}
    \left[
    \sum_{i=1}^{\tau^\textnormal{A}}
    \chi_2 
        \left( 
            Q^{1}_{ \mathbb{I}\{i> \Gamma\} } 
        \Big\|
            Q^{0}_{ \mathbb{I}\{i> \Gamma\} }
        \right)
        \left(
            \beta_i(X^{i-1},Y^{i-1})
        \right)^2
    \right]
\leq
    \delta
.
\end{align}

\begin{remark}
This surrogate metric captures two key aspects of covert sensing: the probing sequence should both use low probing intensity and avoid inducing statistically detectable structure. The $\beta_i^2$ term reflects the instantaneous probing intensity, while the $\chi_2$-divergence quantifies the distinguishability induced at Eve's observation channel.
\end{remark}

\begin{remark}
For $N\to\infty$, this metric takes the form of an expected accumulated cost up to a stopping time, where the instantaneous cost depends only on the current probing probability $\beta_i$ and channel state $\theta_i=\mathbb{I}\{i>\Gamma\}$. This structure makes the resulting active policy optimization amenable to a POMDP formulation.
\end{remark}

\section{Alice's Active Policy as Optimization}
\label{sec:opt_prob}

We formally define Alice's active policy as an optimization problem. First, we give a few definitions on the relevant metrics for covert quickest change detection.

Following the covertness constraint analysis in Sec.~\ref{sec:covert_constr_analys}, we set the following.
\begin{definition}
    The \emph{expected covertness budget (ECB)} of a policy $\Pi^\textnormal{A}$ is 
    \begin{align}
        \label{eq:ecb_def}
        \textnormal{ECB}
            \left(
                \Pi^\textnormal{A}
            \right)
    \triangleq
        \mathbb{E}^{\Pi^\textnormal{A}}
        \left[
            \sum_{i=1}^{\tau^\textnormal{A}}
            \chi_2 
            \left( 
                Q^{1}_{ \mathbb{I}\{i> \Gamma\} } 
            \Big\|
                Q^{0}_{ \mathbb{I}\{i> \Gamma\} }
            \right)
            \left(
                \beta_i(X^{i-1},Y^{i-1}) 
            \right)^2
        \right]
    .
    \end{align}
\end{definition}

As usual in Bayesian quickest change detection (and as in, e.g., ~\cite{banerjee_bayesian_2011,xie2021sequential}), we define the two metrics for quickest change detection as follows.
\begin{definition}
    The \emph{average detection delay (ADD)} of a policy $\Pi^\textnormal{A}$ is
    \begin{align}
        \label{eq:add_def}
        \textnormal{ADD}\left(
                \Pi^\textnormal{A}
            \right)
    \triangleq
        \mathbb{E}^{\Pi^\textnormal{A}}
        \left[
            \left(
                \tau^\textnormal{A}
            -
                \Gamma
            \right)^+
        \right]
    .
    \end{align}
\end{definition}
\begin{definition}
    The \emph{probability of false alarm (PFA)} of a policy $\Pi^\textnormal{A}$ is
    \begin{align}
        \label{eq:pfa_constr}
        \textnormal{PFA}\left(
                \Pi^\textnormal{A}
            \right)
    \triangleq
        \mathbb{P}^{\Pi^\textnormal{A}}
        \left\{
                \tau^\textnormal{A}
            <
                \Gamma
        \right\}
    .
    \end{align}
\end{definition}
We can then formulate Alice's active policy optimization problem as the following. For any pair of given performance metrics $\alpha\in~]0,1[$ and $\delta>0$, Alice designs her active sensing policy $\Pi^\textnormal{A}$ to solve the following problem:
    \begin{align}
        \textnormal{ADD}^\star(\alpha,\delta)
        \triangleq
        ~&\underset{\Pi^\textnormal{A}\in\mathscr{P}^\textnormal{A}}{\inf}
        \quad
        \textnormal{ADD}\left(\Pi^\textnormal{A}\right),
        \\
        &\text{subject to}\quad
        \textnormal{PFA}\left(\Pi^\textnormal{A}\right) \leq \alpha
        \textnormal{~and~}
        \textnormal{ECB}\left(\Pi^\textnormal{A}\right) \leq \delta.
    \end{align}
    where the set $\mathscr{P}^\textnormal{A}$ of admissible policies satisfies the general causal structure detailed in Sec.~\ref{sec:sys_model}.

We address the issue of admissibility by showing that there is always an innocent policy $\Pi^{\boldsymbol{0}}=(\tau^{\boldsymbol{0}},\boldsymbol{0})$ that satisfies both constraints while ensuring a finite ADD. Let the stopping time $\tau^{\boldsymbol{0}}\triangleq N_\alpha\triangleq\lceil|\ln\alpha|/d\rceil$ be a deterministic constant.
Note that
\begin{align}
    \textnormal{PFA}\left(\Pi^{\boldsymbol{0}}\right)
    \triangleq 
    \mathbb{P}_{\pi} \{ \tau^{\boldsymbol{0}} < \Gamma \}
    =
    \mathbb{P}_{\pi} \{ \Gamma > N_\alpha \}
    =
   \exp(- d N_\alpha)
   \leq
   \exp(- d |\ln\alpha| / d)
   = \alpha
   .
\end{align}
Also, $\textnormal{ECB}\left(\Pi^{\boldsymbol{0}}\right)=0\leq\delta$. Moreover, $\textnormal{ADD}\left(\Pi^{\boldsymbol{0}}\right)
\triangleq 
\mathbb{E}_{\pi} \left[
            \left(
                \tau^{\boldsymbol{0}}
            -
                \Gamma
            \right)^+
        \right]
    \le 
        \mathbb{E}_{\pi} \left[
            \tau^{\boldsymbol{0}}
        \right]
    = N_\alpha < \infty$.
    Hence, for all $(\alpha,\delta)$, the minimal ADD satisfies
        \begin{align}
            \textnormal{ADD}^\star(\alpha,\delta)
            \le
            \textnormal{ADD}\left(\Pi^{\boldsymbol{0}}\right)
            \le
            N_\alpha
            =
            \left\lceil \frac{|\ln\alpha|}{d} \right\rceil
            .
        \end{align}
    Therefore, for any $\delta>0$, as $\alpha\to0^+$,
    \begin{align}
        \limsup_{\alpha\to0^+} \frac{\textnormal{ADD}^\star(\alpha,\delta)}{|\ln\alpha|}
        \le
         \limsup_{\alpha\to0^+}
        \frac{\textnormal{ADD}\left(\Pi^{\boldsymbol{0}}\right)}{|\ln\alpha|}
        \le
        \frac{1}{d}
        .
    \end{align}
    On the other hand, since $(N_\alpha-\Gamma)^+ +\Gamma\ge N_\alpha$, we have $\textnormal{ADD}\left(\Pi^{\boldsymbol{0}}\right)+\mathbb{E}_\pi[\Gamma]\ge N_\alpha$, Therefore, we have that $N_\alpha - 1/\rho \le \textnormal{ADD}\left(\Pi^{\boldsymbol{0}}\right)\le N_\alpha$, and thus
    \begin{equation}
        \textnormal{ADD}\left(\Pi^{\boldsymbol{0}}\right)
    =\frac{|\ln\alpha|}{d} + O(1).
    \end{equation}
    This shows that, even without any active probes, Alice can achieve an ADD that scales like $|\ln\alpha|/d$ as $\alpha\to0^+$ solely from her knowledge of the prior. This serves as a baseline, and begs the question: under stringent covertness constraints, can Alice do better than the innocent policy $\Pi^{\boldsymbol{0}}$? If so, how much better?

\section{Second-Order Converse for Covert Bayesian QCD}
\label{sec:sec-ord-conv}

We derive the second-order converse for covert Bayesian QCD by quantifying the tradeoff between information acquisition and covertness. Roughly speaking, the average detection delay is reduced by actively probing the channel and accumulating post-change information, but the ECB constraint fundamentally limits the total probing effort that can be expended covertly.

Let $L_t$ denote the modulated log-likelihood ratio (LLR) of Alice's observation at time $t$, i.e.,
\begin{equation}
    \label{eq:modu_LLR}
    L_t \triangleq L(X_t,Y_t) \triangleq \ln \frac{ P^{X_t}_1 (Y_t) }{  P^{X_t}_0 (Y_t) }
    .
\end{equation}
Let $\ell_t$ denote the unmodulated LLR,
\begin{equation}
    \label{eq:unmodu_LLR}
    \ell_t \triangleq \ell(Y_t) \triangleq \ln \frac{ P^1_1 (Y_t) }{  P^1_0 (Y_t) }
    .
\end{equation}
By the ``no free passive sensing" assumption, $L_t = X_t\ell_t$.

Furthermore, for any policy $\Pi^\textnormal{A}=(\tau^\textnormal{A},\boldsymbol{\beta})$, define the expected post-change sensing effort
\begin{equation}
    B\left(\Pi^\textnormal{A}\right)
    \triangleq
    \mathbb{E}_\pi\left[
        \sum_{t=\Gamma+1}^{\tau^{\textnormal{A}}} \beta_t
    \right],
\end{equation}
which measures the total expected probing activity accumulated after the changepoint and before stopping. We use the convention that $\sum_{i=n}^m(\cdot)\triangleq0$ whenever $m<n$. Observe that
\[
B\left(\Pi^\textnormal{A}\right)
\le
\mathbb{E}_\pi\left[
    \sum_{t=\Gamma+1}^{\tau^\textnormal{A}} 1
\right]
=
\textnormal{ADD}\left(\Pi^\textnormal{A}\right)
<\infty
\]
for any admissible policy.

The following lemma formalizes the intuition that both passive waiting, captured by the prior exponent $d$, and active probing, captured by the information gain $D$, contribute to overcoming the false-alarm constraint.

\begin{lemma}[Non-Asymptotic Information Bound]
\label{lem:info}
Fix any $\alpha\in]0,1[$. For any causal policy $\Pi^\textnormal{A}$ satisfying $\textnormal{PFA}(\Pi^\textnormal{A}) \le \alpha$, its average detection delay $\textnormal{ADD}\left(\Pi^\textnormal{A}\right)$ and expected post-change sensing effort $B\left(\Pi^\textnormal{A}\right)$ satisfy:
\begin{equation}
    d \cdot \textnormal{ADD}\left(\Pi^\textnormal{A}\right) + D \cdot B\left(\Pi^\textnormal{A}\right) \ge |\ln \alpha| - \frac{d}{\rho}
    .
\end{equation}
Recall that we assume $D\triangleq\mathbb{D}(P^1_1\|P^1_0)$ satisfies $0<D<\infty$.
\end{lemma}

\begin{proof}
If $\textnormal{ADD}(\Pi^\textnormal{A})=\infty$, the claim is trivial. Hence assume
$\textnormal{ADD}(\Pi^\textnormal{A})<\infty$. Since
\[
\tau^\textnormal{A} \leq \Gamma + (\tau^\textnormal{A}-\Gamma)^+,
\]
we have
\[
\mathbb{E}_\pi[\tau^\textnormal{A}]\leq \mathbb{E}_\pi[\Gamma]+\textnormal{ADD}(\Pi^\textnormal{A})
= \frac1\rho+\textnormal{ADD}(\Pi^\textnormal{A})<\infty.
\]
In particular, we have $\tau^\textnormal{A}<\infty$ $\mathbb{P}$-almost surely.

We analyze the relative entropy from the true joint probability measure $\mathbb{P}\triangleq\mathbb{P}_\pi$ to a pure-noise reference measure $\mathbb{Q}$. Under $\mathbb{P}$, the changepoint $\Gamma \sim \text{Geo}(\rho)$, and the observations switch channel states at $\Gamma$. Under $\mathbb{Q}$, we maintain $\Gamma \sim \text{Geo}(\rho)$, but the observations are drawn from the pre-change distribution for all $t \ge 1$, regardless of $\Gamma$. Under both measures, the causal policy $\Pi^\textnormal{A}$ operates identically given the observation history. Explicitly, for any $N\in\mathbb{N}$, the joint PMF of $(\Gamma,X^N,Y^N)$ under $\mathbb{P}$ is
\begin{align}
    \label{eq:controlled_sensing_P}
    \mathbb{P}_{\Gamma,X^N,Y^N}(k,x^N,y^N)
    =
    \pi_k
    \prod_{t=1}^N 
        \left[
            \left(\beta_t(x^{t-1},y^{t-1})\right)^{x_t}
            \left( 1- \beta_t(x^{t-1},y^{t-1}) \right)^{1-x_t}
            P^{x_t}_{\theta_t(k)}(y_t)
        \right]
\end{align}
for any $(k,x^N,y^N)\in\mathbb{N}\times\mathcal{X}^N\times\mathcal{Y}^N$, where we recall that $\theta_t(k)\triangleq\mathbb{I}\{t>k\}$, that $\beta_t(x^{t-1},y^{t-1})=0$ for all $t\geq\tau^\textnormal{A}$, and we use the convention $0^0\triangleq1$ here. Similarly, the joint PMF of $(\Gamma,X^N,Y^N)$ under $\mathbb{Q}$ is
\begin{align}
    \label{eq:controlled_sensing_Q}
    \mathbb{Q}_{\Gamma,X^N,Y^N}(k,x^N,y^N)
    =
    \pi_k
    \prod_{t=1}^N 
        \left[
            \left(\beta_t(x^{t-1},y^{t-1})\right)^{x_t}
            \left( 1- \beta_t(x^{t-1},y^{t-1}) \right)^{1-x_t}
            P^{x_t}_0(y_t)
        \right]
,
\end{align}
for any $(k,x^N,y^N)\in\mathbb{N}\times\mathcal{X}^N\times\mathcal{Y}^N$. Note that the information structure here is that of controlled sensing~\cite{Veeravalli_2024}.

First, we evaluate the false alarm probability. Fix any changepoint realization $k\in\mathbb{N}$. Because $\tau^{\textnormal{A}}$ is a stopping time adapted to Alice's natural filtration $\mathcal{F}_t$, the decision to stop at $\tau^{\textnormal{A}} < k$ depends on Alice's observations at most up to $k$, which are drawn from the pre-change distributions under both $\mathbb{P}_k$ and $\mathbb{Q}$, we have $\mathbb{P}_k(\tau^{\textnormal{A}} < k) = \mathbb{Q}(\tau^{\textnormal{A}} < k)$. Hence,
\begin{align}
    \mathbb{P}(\tau^{\textnormal{A}} < \Gamma) 
    &\triangleq
    \sum_{k\in\mathbb{N}} \pi_k \mathbb{P}_k(\tau^{\textnormal{A}} < k)
    \\&= 
    \sum_{k\in\mathbb{N}} \pi_k \mathbb{Q}(\tau^{\textnormal{A}} < k)
    \\
    &= 
    \sum_{k\in\mathbb{N}} \pi_k \mathbb{E}_{\mathbb{Q}}
    \left[ \mathbb{I}(\tau^{\textnormal{A}} < k) \right]
    \\&\overset{\textnormal{(a)}}{=}
    \mathbb{E}_{\mathbb{Q}}
    \left[
    \sum_{k\in\mathbb{N}} \pi_k \mathbb{I}(\tau^{\textnormal{A}} < k)
    \right]
    \\&= 
    \mathbb{E}_{\mathbb{Q}} \left[ \sum_{k\in\mathbb{N}: k>\tau^{\textnormal{A}}} \pi_k \right]
    \\&\overset{\textnormal{(b)}}{=} 
    \mathbb{E}_{\mathbb{Q}} \left[ \exp(-d\tau^\textnormal{A}) \right]
    ,
\end{align}
where (a) follows from Tonelli's theorem, and (b) follows from the definition of geometric prior tail in~\eqref{eq:geo_prior_tail}.
Consequently, the PFA constraint $\textnormal{PFA}(\Pi^\textnormal{A})=\mathbb{P}(\tau^{\textnormal{A}} < \Gamma)\le\alpha$ dictates that
\begin{equation}
    \mathbb{P}(\tau^{\textnormal{A}} < \Gamma) =  \mathbb{E}_{\mathbb{Q}}[\exp(-d\tau^{\textnormal{A}})] \le \alpha
    .
    \label{eq:q_bound}
\end{equation}

Second, we claim that the relative entropy between $\mathbb{P}$ and $\mathbb{Q}$ is exactly the product of the expected post-change unmodulated LLR, $D$, and the expected post-change sensing effort, $B\left(\Pi^\textnormal{A}\right)$, i.e.,
\begin{equation}
    \mathbb{D}(\mathbb{P} \parallel \mathbb{Q}) = D \cdot B\left(\Pi^\textnormal{A}\right).
\end{equation}
Since we assume $D\triangleq\mathbb{D}(P^1_1\|P^1_0)<\infty$, we have $P^1_1\ll P^1_0$. Also, $P^0_1= P^0_0$ from ``no free passive sensing". Hence, we have that $\mathbb{P}\ll\mathbb{Q}$ on every finite-dimensional cylinder set. We first compute the finite-dimensional relative entropy and then pass to the limit, i.e., from~\cite[Theorem~4.11]{polyanskiy2025information}:
\begin{equation}
    \label{eq:relative_entropy_monotone_convergence}
    \mathbb{D}(\mathbb{P} \parallel \mathbb{Q})
    =
    \lim_{N\to\infty} \mathbb{D}(\mathbb{P}_{\mathcal{G}_{N+1}} \parallel \mathbb{Q}_{\mathcal{G}_{N+1}})
    ,
\end{equation}
where for all $t\in\mathbb{N}$ we define the augmented $\sigma$-algebra $\mathcal{G}_t \triangleq \mathcal{F}_t \vee \sigma(\Gamma)$, representing the history of Alice's causal information joined with oracle knowledge of the changepoint, which can be thought of as the information state of a hypothetical ``Genie". Recall $\mathcal{F}_t\triangleq\sigma(X^{t-1},Y^{t-1})$, so $\mathcal{G}_t=\sigma(\Gamma,X^{t-1},Y^{t-1})$.

For any $N\in\mathbb{N}$, from~\eqref{eq:controlled_sensing_P} and~\eqref{eq:controlled_sensing_Q}, the log-likelihood ratio function from $\mathbb{P}$ to $\mathbb{Q}$ of the joint r.v.s $(\Gamma,X^N,Y^N)$ evaluated at $(k,x^N,y^N)$ is
\begin{align}
    \ln \frac{\mathbb{P}_{\Gamma,X^N,Y^N}(k,x^N,y^N)}{\mathbb{Q}_{\Gamma,X^N,Y^N}(k,x^N,y^N)}
    &=
    \sum_{t=1}^{N} \ln \frac{P^{x_t}_{\theta_t(k)}(y_t)}{P^{x_t}_{0}(y_t)}
    \\
    &\overset{\textnormal{(a)}}{=}
    \sum_{t=k+1}^{N} \ln \frac{P^{x_t}_{1}(y_t)}{P^{x_t}_{0}(y_t)}
    \\
    &\overset{\textnormal{(b)}}{=}
    \sum_{t=k+1}^{N} x_t \ell(y_t)
    \\
    &=
    \sum_{t=1}^{N} x_t \ell(y_t) \mathbb{I}\{k<t\} 
    ,
\end{align}
where (a) follows from $\theta_t(k)=\mathbb{I}\{t>k\}$, (b) follows from the definition of LLRs and the ``no free passive sensing" assumption. Then, the log-likelihood ratio r.v. from $\mathbb{P}_{\Gamma,X^N,Y^N}$ to $\mathbb{Q}_{\Gamma,X^N,Y^N}$ is
\begin{align}
    \ln \frac{\mathbb{P}_{\Gamma,X^N,Y^N}}{\mathbb{Q}_{\Gamma,X^N,Y^N}}(\Gamma,X^N,Y^N)
    =
    \sum_{t=1}^{N} X_t \ell(Y_t) \mathbb{I}\{\Gamma<t\} 
     .
\end{align} 
Hence,
\begin{align}
    \mathbb{D}(\mathbb{P}_{\mathcal{G}_{N+1}} \parallel \mathbb{Q}_{\mathcal{G}_{N+1}})
    &=
    \mathbb{E}_{\mathbb{P}} \left[ \sum_{t=1}^{N} X_t \ell(Y_t) \mathbb{I}\{\Gamma<t\}  \right]
    \\
    &\overset{\textnormal{(a)}}{=}
    \sum_{t=1}^{N} 
    \mathbb{E}_{\mathbb{P}} \left[ X_t \ell(Y_t) \mathbb{I}\{\Gamma<t\}  \right]
    \\
    &\overset{\textnormal{(b)}}{=}
    \sum_{t=1}^{N} 
    \mathbb{E}_{\mathbb{P}} \left[ 
        \mathbb{E}_{\mathbb{P}} \left[
            X_t \ell(Y_t) \mathbb{I}\{\Gamma<t\}
            \Big\vert \mathcal{G}_t \right]
        \right]
\end{align}
where (a) is by linearity of expectation and (b) is by the tower rule.

Since $\beta_t=\beta_t(X^{t-1},Y^{t-1})$ is
$\mathcal{F}_t$-measurable and $\mathcal{F}_t\subset\mathcal{G}_t$,
it is also $\mathcal{G}_t$-measurable. Moreover,
$\mathbb{I}\{\Gamma<t\}$ is $\mathcal{G}_t$-measurable. Hence,
\begin{align}
&\mathbb{E}_{\mathbb{P}}
\left[
X_t\ell(Y_t)\mathbb{I}\{\Gamma<t\}
\Bigm| \mathcal{G}_t
\right]
\nonumber\\
&\qquad =
\mathbb{I}\{\Gamma<t\}
\mathbb{P}\{X_t=1\mid \mathcal{G}_t\}
\mathbb{E}_{\mathbb{P}}
\left[
\ell(Y_t)\mid X_t=1,\mathcal{G}_t,\Gamma<t
\right]
\\
&\qquad =
\mathbb{I}\{\Gamma<t\}\beta_t D.
\end{align}
The last equality follows because, on $\{\Gamma<t\}$, the channel is
post-change at time $t$, and under active probing $Y_t\sim P_1^1$,
so that
\begin{equation}
\mathbb{E}_{\mathbb{P}}
\left[
\ell(Y_t)\mid X_t=1,\mathcal{G}_t,\Gamma<t
\right]
=
\mathbb{D}(P_1^1\|P_1^0)=D.
\end{equation}
Therefore, 
\begin{align}
\mathbb{D}(\mathbb{P}_{\mathcal{G}_{N+1}} \parallel \mathbb{Q}_{\mathcal{G}_{N+1}})
&=
D\,
\mathbb{E}_{\mathbb{P}}
\left[
\sum_{t=1}^{N}
\beta_t\mathbb{I}\{\Gamma<t\}
\right].
\end{align}
If we adopt the convention that $\beta_t=0$ after stopping, equivalently
insert $\mathbb{I}\{t\le \tau^\textnormal{A}\}$ in the summand, this becomes
\begin{align}
\mathbb{D}(\mathbb{P}_{\mathcal{G}_{N+1}} \parallel \mathbb{Q}_{\mathcal{G}_{N+1}})
&=
D\,
\mathbb{E}_{\mathbb{P}}
\left[
\sum_{t=1}^{N}
\beta_t\mathbb{I}\{\Gamma<t\}\mathbb{I}\{t\le\tau^\textnormal{A}\}
\right]
\\&=
D\,
\mathbb{E}_{\mathbb{P}}
\left[
\sum_{t=\Gamma+1}^{\min(\tau^\textnormal{A},N)}
\beta_t
\right].
\label{eq:post_change_llr_sum}
\end{align}
Finally, let $N\to\infty$. Since $\tau^\textnormal{A}<\infty$
$\mathbb{P}$-a.s. and the summands in~\eqref{eq:post_change_llr_sum} are nonnegative, the monotone convergence theorem gives
\begin{align}
    \lim_{N\to\infty}
    \mathbb{E}_{\mathbb{P}}
    \left[
    \sum_{t=\Gamma+1}^{\min(\tau^\textnormal{A},N)}
    \beta_t
    \right]
    &=
    \mathbb{E}_{\mathbb{P}}
    \left[
    \sum_{t=\Gamma+1}^{\tau^\textnormal{A}}
    \beta_t
    \right]
    \\
    &\triangleq \label{eq:sum_beta_monotone_convergence}
    B(\Pi^\textnormal{A}).
\end{align}
Combining~\eqref{eq:relative_entropy_monotone_convergence},~\eqref{eq:post_change_llr_sum} and~\eqref{eq:sum_beta_monotone_convergence} yields the desired claim.

Third, we apply the Donsker-Varadhan variational formula~\cite[Thm 4.6]{polyanskiy2025information}, which ensures that for any function $f(X)\in\mathcal{C}_\mathbb{Q}\triangleq\{f:\mathcal{X}\to\mathbb{R}\cup\{-\infty\}:0<\mathbb{E}_{\mathbb{Q}}[\exp{f(X)}]<\infty\}$,
one has
\begin{equation}
    \mathbb{D}(\mathbb{P} \parallel \mathbb{Q})
    =
    \sup_{f\in\mathcal{C}_\mathbb{Q}}
    \mathbb{E}_{\mathbb{P}}[f(X)] - \ln \mathbb{E}_{\mathbb{Q}}[\exp{f(X)}]
    .
\end{equation} 
Choose $f(X) = -d \cdot \tau^{\textnormal{A}} < 0$, so $\mathbb{E}_{\mathbb{Q}}[e^{-d\tau^{\textnormal{A}}}]<1$. On the other hand, since $\mathbb{P}\ll\mathbb{Q}$ and $\mathbb{P}\{\tau^\textnormal{A}<\infty\}=1$, we have
\begin{equation}
    \mathbb{E}_{\mathbb{Q}}[e^{-d\tau^{\textnormal{A}}}] > 0
    ,
\end{equation}
since otherwise $\mathbb{Q}\{\tau^\textnormal{A}<\infty\}=0$, implying $\mathbb{P}\{\tau^\textnormal{A}<\infty\}=0$, a contradiction. 
Hence, $f\in\mathcal{C}_\mathbb{Q}$. Thus, we have
\begin{equation}
    D \cdot B\left(\Pi^\textnormal{A}\right)
    =
    \mathbb{D}(\mathbb{P} \parallel \mathbb{Q})
    \geq
    \mathbb{E}_\mathbb{P}[-d \cdot \tau^{\textnormal{A}}] - \ln \mathbb{E}_{\mathbb{Q}}[e^{-d\cdot \tau^{\textnormal{A}}}]
    .
\end{equation}
Rearranging, we have
\begin{equation}
    \mathbb{E}_\mathbb{P}[-d \cdot \tau^{\textnormal{A}}]
    \leq
    D \cdot B\left(\Pi^\textnormal{A}\right)
    +
    \ln \mathbb{E}_{\mathbb{Q}}[e^{-d\cdot \tau^{\textnormal{A}}}]
    .
\end{equation}
Substituting the upper bound from \eqref{eq:q_bound} yields:
\begin{equation}
    -d \cdot
    \mathbb{E}_\mathbb{P}[\tau^{\textnormal{A}}] \le D \cdot B\left(\Pi^\textnormal{A}\right) + \ln(\alpha) = D \cdot B\left(\Pi^\textnormal{A}\right) - |\ln \alpha| 
\end{equation}
Rearranging gives $d \cdot \mathbb{E}_\mathbb{P}[\tau^{\textnormal{A}}] + D \cdot B\left(\Pi^\textnormal{A}\right) \ge |\ln \alpha| $. Decomposing $\mathbb{E}_\mathbb{P}[\tau^{\textnormal{A}}] = \mathbb{E}_\mathbb{P}[\min(\tau^{\textnormal{A}}, \Gamma)] + \textnormal{ADD}(\Pi^\textnormal{A})$, and applying the bound $\mathbb{E}_\mathbb{P}[\min(\tau^{\textnormal{A}}, \Gamma)] \le \mathbb{E}_\mathbb{P}[\Gamma] = 1/\rho$, we obtain:
\begin{equation}
    d \cdot \textnormal{ADD}(\Pi^\textnormal{A}) + \frac{d}{\rho} + D \cdot B\left(\Pi^\textnormal{A}\right) \ge |\ln \alpha|
\end{equation}
as desired.
\end{proof}

\begin{lemma}[Covertness Bottleneck]
\label{lem:covert_bottle}
Fix any $\delta>0$. For any causal policy $\Pi^\textnormal{A}$ satisfying $\mathrm{ECB}(\Pi^\textnormal{A}) \le \delta$, the expected post-change sensing effort $B\left(\Pi^\textnormal{A}\right)$ is upper bounded as
\begin{equation}
    B\left(\Pi^\textnormal{A}\right) \le \sqrt{\frac{\delta}{\chi_{2,1}} \textnormal{ADD}\left(\Pi^\textnormal{A}\right)}
\end{equation}
\end{lemma}

\begin{proof}
Notice that the ECB definition in~\eqref{eq:ecb_def} can be written compactly as
\begin{align}
        \textnormal{ECB}
            \left(
                \Pi^\textnormal{A}
            \right)
    &\triangleq
        \mathbb{E}^{\Pi^\textnormal{A}}
        \left[
            \sum_{i=1}^{\tau^\textnormal{A}}
            \chi_2 
            \left( 
                Q^{1}_{ \mathbb{I}\{i> \Gamma\} } 
            \Big\|
                Q^{0}_{ \mathbb{I}\{i> \Gamma\} }
            \right)
            \left(
                \beta_i(X^{i-1},Y^{i-1}) 
            \right)^2
        \right]
    \\&=
        \mathbb{E}^{\Pi^\textnormal{A}}
        \left[
            \sum_{i=1}^{\tau^\textnormal{A}}
            \left(
                \chi_{2,0}
                \mathbb{I}\{i\le \Gamma\}
                +
                \chi_{2,1}
                \mathbb{I}\{i>\Gamma\}
            \right)
                \beta_i^2
        \right]
    \\&=
    \chi_{2,0}
        \mathbb{E}^{\Pi^\textnormal{A}}
        \left[
            \sum_{i=1}^{\min(\tau^\textnormal{A},\Gamma)}
            \beta_i^2
        \right]
    +
    \chi_{2,1}
        \mathbb{E}^{\Pi^\textnormal{A}}
        \left[
            \sum_{i=\Gamma+1}^{\tau^\textnormal{A}}
            \beta_i^2
        \right]
    .
\end{align}
Hence, for any causal policy $\Pi^\textnormal{A}$ with $\mathrm{ECB}(\Pi^\textnormal{A}) \le \delta$, dropping the pre-change part of ECB (since we know as $\alpha\to0^+$, the post-change part dominates as the policy takes sufficiently longer time in the post-change regime to meet the strict PFA requirement) yields, recalling that $\chi_{2,1}>0$ by assumption, the bound
\begin{align}
    \mathbb{E}^{\Pi^\textnormal{A}}
        \left[
            \sum_{i=\Gamma+1}^{\tau^\textnormal{A}}
            \beta_i^2
        \right]
    \leq
    \frac{\delta
        }{\chi_{2,1}}
    .
\end{align}
On the other hand, from Cauchy-Schwarz inequality,
\begin{align}
    B\left(\Pi^\textnormal{A}\right)
        =
    \mathbb{E}^{\Pi^\textnormal{A}}
        \left[
            \sum_{i=\Gamma+1}^{\tau^\textnormal{A}}
            \beta_i
        \right]
    &\leq
    \sqrt{
    \mathbb{E}^{\Pi^\textnormal{A}}
        \left[
            \sum_{i=\Gamma+1}^{\tau^\textnormal{A}}
            \beta_i^2
        \right]
    \cdot
    \mathbb{E}^{\Pi^\textnormal{A}}
        \left[
            \sum_{i=\Gamma+1}^{\tau^\textnormal{A}}
            1
        \right]
    }
    .
\end{align}
Combining both inequalities yields the desired result. 
\end{proof}

\begin{remark}
    The equality condition above provides a hint for how to allocate the probing probabilities: if possible (e.g., with oracle knowledge of the changepoint), Alice should try to allocate a flat sequence between $\Gamma+1$ and $\tau^\textnormal{A}$. This provides some intuition for the constant-sensing-probability achievability scheme below. 
\end{remark}

Let us provide the following lemma, which is used in proving the second-order asymptotic converse.
\begin{lemma}
    \label{lem:sqrt(1+x)_lower_bound}
    Choose any $r\in~]0,1[~$, and let 
\begin{equation}
     A_r 
     \triangleq
    \frac{1}{8}(1-r)^{-3/2}
    > 0
    .
\end{equation}
Then for all $x\in[-r,r]$,
\begin{equation}
    \sqrt{1+x} 
    \geq 
    1 
    + \frac{x}{2}
    - A_r x^2
    .
\end{equation}
\end{lemma}
\begin{proof}
    Fix any $r\in~]0,1[~$.
    Let $f(x)=\sqrt{1+x}$ for $x\in[-r,r]$.
    Differentiating, we have
    \begin{align}
        f'(x) = \frac{1}{2\sqrt{1+x}},&&f''(x)=-\frac{1}{4(1+x)^{3/2}}.
    \end{align}
    Applying Taylor theorem with Lagrange remainder around $x=0$, we have that for every $x\in[-r,r]$,
    \begin{align}
        \sqrt{1+x} 
        = 1 + \frac{1}{2}x + \frac{f''(\xi_x)}{2!}x^2
        = 1 + \frac{1}{2}x - \frac{1}{8(1+\xi_x)^{3/2}}x^2
    \end{align}
    for some $\xi_x$ between $0$ and $x$. Then $\xi_x\in[-r,r]$ also, thus $1+\xi_x\ge1-r>0$. Hence,
    \begin{align}
        \frac{1}{8(1+\xi_x)^{3/2}} \le \frac{1}{8(1-r)^{3/2}}
        .
    \end{align}
    As $x^2\ge0$, we obtain that, for any $x\in[-r,r]$,
    \begin{align}
        \sqrt{1+x} 
        \ge
        1 + \frac{1}{2}x - \frac{1}{8(1-r)^{3/2}}x^2
        ,
    \end{align}
    completing the proof.
\end{proof}

\begin{theorem}[Maximum Covert Sensing Gain]
    \label{thm:second-order-converse}
As $\alpha \to 0^+$, the average detection delay of any admissible causal policy $\Pi^\textnormal{A}$ is lower-bounded by:
\begin{equation}
    \textnormal{ADD}\left(\Pi^\textnormal{A}\right) \ge \frac{|\ln \alpha|}{d} - \frac{D \sqrt{\delta}}{d^{3/2} \sqrt{\chi_{2,1}}} \sqrt{|\ln \alpha|} + O(1)
    .
\end{equation}
The subtracted term of order $\sqrt{|\ln \alpha|}$ represents the asymptotic maximum covert sensing gain possible under the budget $\delta>0$.
\end{theorem}

\begin{proof}
Substituting Lemma~\ref{lem:covert_bottle} into Lemma~\ref{lem:info} yields an exact non-asymptotic quadratic inequality in terms of $\sqrt{\textnormal{ADD}(\Pi^\textnormal{A})}$:
\begin{equation}
    d \cdot \textnormal{ADD}(\Pi^\textnormal{A}) + D \sqrt{\frac{\delta}{\chi_{2,1}}} \sqrt{\textnormal{ADD}(\Pi^\textnormal{A})} \ge |\ln \alpha| - C_1
\end{equation}
where $C_1 = \frac{d}{\rho} >0$. 

Letting $x = \sqrt{\textnormal{ADD}(\Pi^\textnormal{A})}\geq 0$ and $K = D \sqrt{\frac{\delta}{\chi_{2,1}}}>0$, we have the equivalent inequality
\begin{equation}
    f(x)
    \triangleq
    d \cdot x^2 + K \cdot x
    - (|\ln \alpha| - C_1)
    \geq 0
    .
\end{equation} 
Since $d>0$, the graph $y=f(x)$ is an upward parabola. The discriminant of $f(x)$ is
\begin{equation}
    \Delta_{f} = K^2 + 4d(|\ln \alpha| - C_1)
    .
\end{equation}
Restricting $\alpha$ such that $0<\alpha<\alpha_0\triangleq\exp(-C_1)<1$ onwards, we have
\begin{align}
    \Delta_f > K^2 
    .
\end{align}
By the quadratic formula, $f(x)$ has two roots $r_{\pm}$ as
\begin{equation}
    r_\pm
    =
    \frac{-K \pm \sqrt{ \Delta_f }}{2d}
    ,
\end{equation}
where $r_+>0$ and $r_-<0$. Since $x\geq0$, the resulting equivalent inequality from $f(x)\geq 0$ is
\begin{equation}
    x 
    \ge 
    r_+
    =
    \frac{-K + \sqrt{K^2 + 4d(|\ln \alpha| - C_1)}}{2d}
\end{equation}
Factoring out $4d|\ln \alpha|$ from the discriminant gives:
\begin{equation}
    \label{eq:sqrt_add_lb_before_Taylor_expand}
    x \ge \frac{-K + 2\sqrt{d|\ln \alpha|} \sqrt{1 + \frac{K^2 - 4d C_1}{4d|\ln \alpha|}}}{2d}
\end{equation}

Fix any $r\in~]0,1[~$.
Choose $\alpha_r>0$ such that for all $0<\alpha<\alpha_r$, we have
\begin{equation}
    \frac{|K^2 - 4d C_1|}{4d|\ln \alpha|}
    \leq r
    .
\end{equation}
Considering only $\alpha$ such that $0<\alpha<\min\{\alpha_0,\alpha_r\}$ onwards, by Lemma~\ref{lem:sqrt(1+x)_lower_bound} we have 
\begin{equation}
    \sqrt{1 + \frac{K^2 - 4d C_1}{4d|\ln \alpha|}}
    \geq
    1
    +
    \frac{K^2 - 4d C_1}{8d|\ln \alpha|}
    -
    A_r
    \left(\frac{K^2 - 4d C_1}{4d}\right)^2
    \frac{1}{|\ln \alpha|^2}
    .
\end{equation}
Applying this into~\eqref{eq:sqrt_add_lb_before_Taylor_expand}, we have:
\begin{align}
x
&\ge
\sqrt{\frac{|\ln\alpha|}{d}}
-\frac{K}{2d}
+
\frac{K^2-4dC_1}{8d^{3/2}\sqrt{|\ln\alpha|}}
-
\frac{A_r(K^2-4dC_1)^2}{16d^{5/2}|\ln\alpha|^{3/2}}.
\end{align}
Collecting terms based on orders of $\sqrt{|\ln\alpha|}$, we have:
\begin{equation}
    \label{eq:sqrt_ADD_LB_after_Taylor}
    x \ge \sqrt{\frac{|\ln \alpha|}{d}} - \frac{K}{2d} + O\left(\frac{1}{\sqrt{|\ln \alpha|}}\right)
    ,
\end{equation}
which goes to $\infty$ as $\alpha\to0^+$. Hence, there is an $\alpha_1>0$ such that the right-hand side of~\eqref{eq:sqrt_ADD_LB_after_Taylor} is strictly positive for all $0<\alpha<\alpha_1$.

Hence, for all $\alpha$ such that $0<\alpha<\min\{\alpha_0,\alpha_r,\alpha_1\}$, we can square both sides of~\eqref{eq:sqrt_ADD_LB_after_Taylor} to obtain:
\begin{equation}
    \textnormal{ADD}(\Pi^\textnormal{A})
    =
    x^2
    \ge \frac{|\ln \alpha|}{d} - \frac{K}{d^{3/2}} \sqrt{|\ln \alpha|} + \frac{K^2}{4d^2} + O(1)
\end{equation}
The term $\frac{K^2}{4d^2}$ is a strictly positive constant independent of $\alpha$, which is absorbed into the $O(1)$ term. Substituting $K = D \sqrt{\frac{\delta}{\chi_{2,1}}}$ reveals the maximum covert sensing gain, concluding the proof.
\end{proof}

\section{Achievability Scheme and Its Second-Order Asymptotics}
\label{sec:achievability}

In this section, we construct an admissible covert quickest change detection policy $\Pi^\textnormal{A} = (\tau^{\textnormal{A}}, \beta_\alpha^*)$ with a constant sensing probability schedule. We analyze its performance using non-asymptotic renewal theory and establish its second-order asymptotics, showing that it matches the second-order converse of $\textnormal{ADD}$ versus $|\ln\alpha|$ in Theorem~\ref{thm:second-order-converse}.

Throughout this section, when we say that the policy has constant sensing probability $\beta$, we mean that Alice uses sensing probability $\beta$ while the procedure is running and sends the innocent symbol after stopping. Equivalently, the effective sensing probability is $\beta$ before the stopping time and $0$ on and after the stopping time. This convention changes sums at the stopping epoch by at most one term and hence does not affect the second-order asymptotics.

\subsection{Preliminaries and the Proposed Policy}

\begin{definition}[The Constant-Sensing-Probability Shiryaev-type Policy]
\label{def:proposed_policy}
Alice employs a constant sensing probability sequence $\beta_t = \beta_\alpha^* \in ~]0, 1]$ for all $t \ge 1$. The exact value of $\beta_\alpha^*$ is defined in Def.~\ref{def:optimal_sensing_rate}. The stopping time $\tau^{\textnormal{A}}$ is the classical Shiryaev rule, defined as the first time the posterior probability of the change exceeds $1-\alpha$:
\begin{equation}
    \tau^{\textnormal{A}} \triangleq \inf \left\{ t \ge 1 : 
    p_t \triangleq \mathbb{P}^{\Pi^\textnormal{A}}(\Gamma < t \mid \mathcal{F}_t) \ge 1 - \alpha \right\}
\end{equation}
where $\mathcal{F}_t$ is the filtration generated by the observations up to time $t$. Equivalently, this is a thresholding rule on the posterior odds $\Lambda_t$:
\begin{equation}
    \tau^{\textnormal{A}} = \inf \left\{ t \ge 1 : \Lambda_t
    \triangleq \frac{p_t}{1-p_t} \ge \frac{1-\alpha}{\alpha} \right\}.
\end{equation}
\end{definition}

A few remarks are in order.
\begin{itemize}
    \item
        The constant-sensing-probability description refers to the policy's sensing phase. Operationally, once the stopping time $\tau^\textnormal{A}$ occurs, Alice declares that the change has occurred and subsequently sends the innocent symbol. Equivalently, in sums involving sensing costs, one may view the effective sensing probability as equal to $\beta_\alpha^*$ before stopping and zero after stopping. This convention is consistent with the ECB definition and does not affect the asymptotic analysis.
    \item
        Because Alice employs a constant sensing probability while the procedure is running, the underlying pre-stopping probing actions may be viewed as independent and identically distributed (i.i.d.) $\mathrm{Ber}(\beta)$. Since the channel $W_{Y,Z|X,\theta}$ is a DMC and hence so is $P_{Y|X,\theta}\triangleq P^X_\theta(Y)$, this implies that the joint distribution of $(X_t,Y_t)$ factors as a product post-change ($\theta=1$), further implying that $L_t=L(X_t,Y_t)$ is i.i.d. post-change, i.e., $t>\Gamma$. Similarly for the pre-change ($\theta=0$) duration. 
    \item 
        Following~\cite[eq.~(2.5)]{tartakovsky_general_2005}, for all $i,n\in\mathbb{N}$ with $i\le n$, define the cumulative modulated LLR by
        \begin{equation}
            \label{eq:cumu_llr}
                V_{[i:n]} \triangleq \sum_{t=i}^n L_t
                ,
        \end{equation}
        which represents the LLR at time $n$ for the hypothesis that change occurred at time $t=i-1$ versus at $t=\infty$ (no change),
        after the index shift $\lambda = \Gamma + 1$ versus ~\cite{tartakovsky_general_2005} induced by our changepoint convention.
        The following recursion is equivalent to the classical  Shiryaev posterior-odds recursion in~\cite[eq.~(2.9)]{tartakovsky_general_2005}.
        For all $n\in\mathbb{N}$, we have
        \begin{equation}
            \label{eq:cumu_LLR_expr}
            \Lambda_n = T_n^{-1} \sum_{m=1}^n \pi_{m} \exp(V_{[m+1:n]})
            .
        \end{equation}
    \item
        Following the standard argument for Shiryaev-type stopping time, namely, a single-thresholding on the Bayesian posterior of change, it can be shown via the same argument as those around~\cite[eq.~(2.4)]{tartakovsky_general_2005} that
            \begin{equation}
                \label{eq:PFA_satisfaction}
                    \textnormal{PFA}\left( \Pi^\textnormal{A} \right)
                    \le 1-(1-\alpha) = \alpha 
                    ,
            \end{equation}
            so Alice's constant-sensing-probability Shiryaev policy $\Pi^\textnormal{A}$ satisfies the PFA constraint, at least for all $\alpha\in~]0,1/2[~$.
\end{itemize}

\subsection{ADD Upper Bound via Lorden's Theorem}

By establishing the i.i.d. nature of the post-change environment, we bypass asymptotic limit arguments and bound the ADD using exact non-asymptotic renewal theory, in particular we use~\cite{lorden_excess_1970}.

\begin{theorem}[Non-Asymptotic ADD Bound]
\label{thm:add_bound}
Fix any $\alpha\in~]0,1/2[~$.
For the constant-sensing-probability Shiryaev-type policy $\Pi^\textnormal{A}$ with $\beta \in ~]0, 1]$, the ADD is bounded from above as:
\begin{equation}
    \textnormal{ADD}\left(\Pi^\textnormal{A}\right)
    \le 
    \frac{|\ln \alpha| - |\ln(1-\alpha)|+ C_{\rho} }{\beta D + d} + \frac{\beta V + 2d\beta D + d^2}{(\beta D + d)^2}
    \triangleq
    \textnormal{ADD}_\textnormal{upper}
    ,
\end{equation}
where $C_{\rho}\triangleq \ln((1-\rho)/\rho)$ is a constant that depends on the prior parameter $\rho$ only.
\end{theorem}

\begin{proof}
    Following the ideas in~\cite[eq.~(3.37)-(3.39)]{tartakovsky_general_2005}, from~\eqref{eq:cumu_LLR_expr} we obtain that, for any $n,k\in\mathbb{N}$ with $n>k$, 
\begin{align}
    \Lambda_n 
    \geq
        T_n^{-1} \pi_k \exp(V_{[k+1:n]}) 
    ,
\end{align}
after dropping all other terms in the sum, all of which are non-negative.
Taking the logarithm, we have
\begin{align}
    \ln \Lambda_n 
        &\ge 
            \ln \frac{\pi_k}{T_n} + V_{[k+1:n]} 
        \\&=
            \ln \frac{\rho(1-\rho)^{k-1}}{(1-\rho)^n} + V_{[k+1:n]} 
        \\&=
            -(n+1-k) \ln(1-\rho) + \ln\rho + V_{[k+1:n]}
        \\&=
            (n-k) d + (d + \ln\rho) + V_{[k+1:n]}
        ,
\end{align}
after explicit calculations. Let the prior-dependent constant $C_{\rho}\triangleq-(d + \ln\rho)=\ln((1-\rho)/\rho)$.

Applying the change of variable $t=n-k$, representing a shift by the changepoint candidate $k$, where the condition $n > k$ becomes $t\ge 1$, we have
\begin{equation}
    \label{eq:log-pos-odds-comp}
    \ln \Lambda_{k+t} + C_{\rho} 
    \ge 
    t d + V_{[k+1:k+t]} 
    = 
    \sum_{m=k+1}^{k+t} ( L_{m} + d )
    \triangleq
    S_{[k+1:k+t]}
    .
\end{equation}

The constant-sensing-probability Shiryaev's stopping time can be written as
\begin{equation}
    \tau^\textnormal{A}
    =
    \inf \left\{ n \ge 1 : \ln \Lambda_{n} + C_{\rho}
   \ge b_\alpha + C_{\rho}
   \right\}
\end{equation}
where $b_\alpha\triangleq \ln((1-\alpha)/\alpha) = |\ln\alpha|-|\ln(1-\alpha)|$. Then on the event $\{\tau^\textnormal{A} > k\}$, we have
\begin{equation}
    \tau^\textnormal{A} - k
    =
    \inf \left\{ t \ge 1 : \ln \Lambda_{t+k} + C_{\rho}
        \ge b_\alpha + C_{\rho}
   \right\}
   .
\end{equation}
Define, for each changepoint candidate $k$, the stopping time
\begin{equation}
    \eta_{\alpha}(k) 
    \triangleq 
    \inf 
    \left\{ t \ge 1 : S_{[k+1:k+t]}
        \ge b_\alpha + C_{\rho}
    \right\}
    .
\end{equation}
Then from~\eqref{eq:log-pos-odds-comp}, we have that
$\ln \Lambda_{k+t} + C_{\rho} 
    \ge 
    S_{[k+1:k+t]}$, and it follows
that on $\{\tau^\textnormal{A} > k\}$,
\begin{equation}
    \tau^\textnormal{A} - k \le \eta_{\alpha}(k)
    .
\end{equation}
Hence, the conditional ADD on changepoint being $k$ is upper bounded as
\begin{equation}
    \label{eq:CADD_UB}
    \mathbb{E}_k \left[
    \left( \tau^\textnormal{A} - k \right)^+
    \right]
    \le
    \mathbb{E}_k \left[
    \eta_{\alpha}(k) \mathbb{I}\{\tau^\textnormal{A} \ge k\}
    \right]
    \le
    \mathbb{E}_k \left[
    \eta_{\alpha}(k)
    \right]
    .
\end{equation}

Define the step variable $U_m \triangleq L_{m} + d$ for $m\in\mathbb{N}^*$. Under the measure $\mathbb{P}_k$, for all $m > k$, $U_m$ is i.i.d. with a strictly positive mean:
\begin{equation}
    \mathbb{E}_k[
        U_m
        ] = 
        \mathbb{E}_k[L_{m}] + d = \beta D + d \ge d > 0
    .
\end{equation}
Let $B_\alpha \triangleq b_\alpha + C_{\rho}$. Define the overshoot $R_{\alpha,k}$ of the i.i.d. random walk
\begin{equation}
    S_{[k+1:k+t]}
    \triangleq
    \sum_{m=k+1}^{k+t} U_m
\end{equation}
over the boundary $B_\alpha$, namely,
\begin{equation}
    R_{\alpha,k} \triangleq S_{[k+1:k+\eta_{\alpha}(k)]} - B_\alpha \ge 0
    .
\end{equation}
Hence,
\begin{align}
    B_\alpha + \mathbb{E}_k \left[ R_{\alpha,k} \right] 
=
    \mathbb{E}_k \left[ S_{[k+1:k+\eta_{\alpha}(k)]} \right]
=   \mathbb{E}_k \left[ \sum_{m=k+1}^{k+\eta_{\alpha}(k)} U_m \right]
=   \mathbb{E}_k[ \eta_\alpha(k) ] \mathbb{E}_k[ U_{k+1} ] 
    ,
\end{align}
where in the last step Wald's Equation~\cite{wald_cumulative_1944} applies because $U_t$ is i.i.d. post-change ($\forall t > k$) with $\mathbb{E}_k[|U_t|]=\mathbb{E}_k[|L_t+d|]\le\mathbb{E}_k[|L_t|]+d=\mathbb{E}_k[X_t|\ell_t|]+d\le\mathbb{E}_k[|\ell_t|]+d<\sqrt{V}+d<\infty$, and $\eta_\alpha(k)$ is a stopping time with $\mathbb{E}_k[\eta_\alpha(k)]<\infty$ as it is the first-hitting time of an i.i.d. random walk with strictly positive drift.
Rearranging, we obtain
\begin{equation}
    \label{eq:expect_eta}
    \mathbb{E}_k[ \eta_\alpha(k) ]
    =
    \frac{ B_\alpha + \mathbb{E}_k \left[ R_{\alpha,k} \right] }{\mathbb{E}_k[ U_{k+1} ] }
    .
\end{equation}

By Lorden's inequality for excess over the boundary~\cite{lorden_excess_1970}, the expected overshoot for an i.i.d. random walk with a positive mean is bounded by the ratio of its (uncentered) second moment to its first moment:
\begin{equation}
    \mathbb{E}_k[ R_{\alpha,k} ] 
    \le \frac{\mathbb{E}_k[U_{k+1}^{2}]}{\mathbb{E}_k[ U_{k+1} ]}
    .
\end{equation}
We note that no extra assumption on the (non-)arithmeticity of the r.v.s $S_k$ are needed.

Substituting this overshoot bound into \eqref{eq:expect_eta} yields:
\begin{equation}
    \label{eq:expect_eta_UB}
    \mathbb{E}_k[ \eta_\alpha(k) ]
    \le
    \frac{ B_\alpha }{\mathbb{E}_k[ U_{k+1} ] }
    +
    \frac{\mathbb{E}_k[ U_{k+1}^2 ]}{ \left( \mathbb{E}_k[ U_{k+1} ] \right)^2}
    .
\end{equation}
Evaluating the uncentered second moment yields:
\begin{equation}
    \mathbb{E}_k[ U_{k+1}^2] = \mathbb{E}_k[( L_{k+1} + d)^2] = \mathbb{E}_k[ L_{k+1}^2] + 2d\mathbb{E}_k[ L_{k+1} ] + d^2 = \beta V + 2 d \beta D + d^2
    .
\end{equation}
Substituting these moments into the fractions in~\eqref{eq:expect_eta_UB}, and averaging~\eqref{eq:CADD_UB} over the prior yields
\begin{equation}
    \textnormal{ADD}\left( \Pi^\textnormal{A} \right)
    =
    \sum_{k\in\mathbb{N}} \pi_k \mathbb{E}_k \left[
    \left( \tau^\textnormal{A} - k \right)^+
    \right]
    \le
    \frac{|\ln \alpha| - |\ln(1-\alpha)| + C_{\rho}}{\beta D + d} + \frac{\beta V + 2d\beta D + d^2}{(\beta D + d)^2}
.
\end{equation}

\end{proof}

\begin{remark}
    We would like to emphasize that this upper bound on $\textnormal{ADD}\left( \Pi^\textnormal{A} \right)$ for constant-sensing-probability Shiryaev-type policies is non-asymptotic in $\alpha$, in the sense that it works for at least all $\alpha\in~]0,1/2[~$. This bound allows us to rigorously upper bound the ECB constraint of these types of policies with a further relaxed upper bound in terms of $\alpha$, enabling the construction of a certain sensing probability $\beta_\alpha^*$ that varies with $\alpha$ to satisfy the ECB for the whole range of $\alpha\in~]0,1/2[~$.  
\end{remark}
\begin{remark}
    It is straightforward to show that the upper bound we obtain in Theorem~\ref{thm:add_bound}, $\textnormal{ADD}_\textnormal{upper}$, is strictly positive at least for every $\alpha\in~]0,1/(1+\rho)[~$, and, in particular, for every $\alpha\in~]0,1/2[~$. To see this, observe that $\textnormal{ADD}_\textnormal{upper}$ can be expressed as
    \begin{align}
        \textnormal{ADD}_\textnormal{upper}
        &=
        \frac{ \ln\frac{1-\alpha}{\alpha} + C_\rho }{\beta D + d} 
        + 
        1
        +
        \frac{\beta V - \beta^2 D^2}{(\beta D + d)^2}
        \\&=
        \frac{ \ln\frac{1-\alpha}{\alpha} + \ln\frac{1-\rho}{\rho} + d }{\beta D + d} 
        + 
        \frac{\beta D}{\beta D+d}
        +
        \frac{\beta V - \beta^2 D^2}{(\beta D + d)^2}
        .
    \end{align}
    Since $V\ge D^2$ by Cauchy-Schwarz and $\beta\in[0,1]$, we have
    \begin{align}
        \beta V - \beta^2 D^2 
        =
        \beta (V - \beta D^2)
        \ge
        \beta (V - D^2)
        \ge
        0
        .
    \end{align}
    Hence, to ensure $ \textnormal{ADD}_\textnormal{upper}>0$, it suffices to ensure that
    \begin{align}
        \ln\frac{1-\alpha}{\alpha} + \ln\frac{1-\rho}{\rho} + d > 0
        .
    \end{align}
    Using $d=-\ln(1-\rho)$, this condition is equivalent to
    \begin{align}
        \ln\frac{1-\alpha}{\alpha} - \ln \rho > 0
        \Leftrightarrow
        \frac{1-\alpha}{\alpha} > \rho
        \Leftrightarrow
        \alpha < \frac{1}{1+\rho}
        .
    \end{align}
    Since the prior parameter $\rho\in~]0,1[~$, in particular for all $\alpha\in~]0,1/2[~$, this condition is ensured, and thus $\textnormal{ADD}_\textnormal{upper}>0$.
\end{remark}

\subsection{Proposed Sensing Rate and Its Second-Order Asymptotics}

To ensure satisfaction of the ECB constraint while permitting a closed-form solution for the sensing rate, we construct a structurally relaxed upper bound for the ADD. 

The non-asymptotic ADD upper bound $\textnormal{ADD}_\textnormal{upper}$ from Theorem~\ref{thm:add_bound} can also be written as
\begin{equation}
    \textnormal{ADD}_\textnormal{upper}
    =
    \frac{ \ln\frac{1-\alpha}{\alpha} - \ln\frac{\rho}{1-\rho} }{\beta D + d} + \frac{\beta V + 2d\beta D + d^2}{(\beta D + d)^2}
    .
\end{equation}
Since $x\mapsto\ln((1-x)/x)$ is strictly decreasing on $]0,1[$, restricting $\alpha\in~]0,1-\rho]$ ensures that the nominator of the first term is non-negative. Because $\beta D \ge 0$, we have
\begin{equation}
    \frac{ \ln\frac{1-\alpha}{\alpha} - \ln\frac{\rho}{1-\rho} }{\beta D + d}
    \le
    \frac{ \ln\frac{1-\alpha}{\alpha} - \ln\frac{\rho}{1-\rho} }{d}
    =
    \frac{|\ln \alpha| - |\ln(1-\alpha)| + C_{\rho}}{d}
    .
\end{equation}

Furthermore, the overshoot term in Theorem \ref{thm:add_bound} can be further upper bounded as follows. We have
\begin{align}
    \frac{\beta V + 2d\beta D + d^2}{(\beta D + d)^2}
    &\overset{\textnormal{(a)}}{\le}
    \frac{\beta V + 2d\beta D + d^2}{d^2}
    \\&=
    \frac{\beta (V + 2dD) + d^2}{d^2}
    \\&\overset{\textnormal{(b)}}{\le}
    \frac{ V + 2dD + d^2}{d^2}
    \triangleq
    M_{\textnormal{over}}
    ,
\end{align}
where (a) is because $\beta D \ge 0$, and (b) is because $\beta\in[0,1]$ and $V+2dD>0$. Therefore, the expected overshoot is universally bounded by the constant $ M_{\textnormal{over}}$.

Correspondingly, the relaxed upper bound for ADD of any constant-sensing-probability Shiryaev-type policy $\Pi^\textnormal{A}$ is:
\begin{equation}
    \textnormal{ADD}\left( \Pi^\textnormal{A} \right)
    \le
    \textnormal{ADD}_\textnormal{relaxed} \triangleq \frac{|\ln \alpha| - |\ln(1-\alpha)| + C_{\rho}}{d} + M_{\textnormal{over}}
    ,
\end{equation}
for all $\alpha\in~]0,\min(1/2,1-\rho)[$.

By decoupling the post-change and pre-change costs for any such policy $\Pi^\textnormal{A}=(\tau^\textnormal{A},\beta)$, the ECB can be upper bounded as
\begin{align}
    \textnormal{ECB}\left( \Pi^\textnormal{A} \right) 
    &\le
    \beta^2 \left( \chi_{2,0} \mathbb{E}[\min(\tau^{\textnormal{A}}, \Gamma)] + \chi_{2,1} \textnormal{ADD}\left( \Pi^\textnormal{A} \right) \right)
    \\
    &\leq
    \beta^2 \left( \chi_{2,0} \mathbb{E}[\Gamma] + \chi_{2,1} \textnormal{ADD}_\textnormal{upper} \right)
    \\
    &\leq\label{eq:ECB_upper_relaxed}
    \beta^2 \left( \frac{ \chi_{2,0}}{\rho} + \chi_{2,1} \textnormal{ADD}_\textnormal{relaxed} \right)
\end{align}
after bounding the pre-change duration by the expected prior $\mathbb{E}[\Gamma] = 1/\rho$ and ADD by $\textnormal{ADD}_\textnormal{relaxed}$, on the range $\alpha\in~]0,\min(1/2,1-\rho)[$. The first inequality is written to be robust to the precise stopping-time indexing convention: after stopping, Alice sends the innocent symbol, so any boundary discrepancy at the stopping epoch can only reduce the effective sensing cost relative to the displayed upper bound.

We can thus enforce the ECB constraint by setting $\beta_\alpha^*$ as follows.
\begin{definition}[Proposed Sensing Rate]
\label{def:optimal_sensing_rate}
We define the unconstrained sensing rate $\tilde{\beta}_\alpha$ by equating the upper bound of the ECB to the total budget $\delta$:
\begin{equation}
    \tilde{\beta}_\alpha \triangleq \sqrt{ \frac{\delta}{ \frac{\chi_{2,0}}{\rho} + \chi_{2,1} \left( \frac{|\ln \alpha| - |\ln(1-\alpha)| + C_{\rho}}{d} + M_{\textnormal{over}} \right) } }
    .
\end{equation}
Because probabilities are bounded, the admissible sensing rate is defined as $\beta_\alpha^* \triangleq \min(1, \tilde{\beta}_\alpha)$.
\end{definition}

\begin{theorem}[Second-Order Asymptotics]
\label{thm:second_order_asymptotics}
For any fixed covertness budget $\delta > 0$, as the false alarm bound $\alpha \to 0^+$, the constant-sensing-probability Shiryaev-type policy $\Pi^\textnormal{A} = (\tau^{\textnormal{A}}, \beta_\alpha^*)$ satisfies $\textnormal{ECB}(\Pi^\textnormal{A}) \le \delta$ for all $\alpha$ small enough, and achieves the following second-order asymptotic expansion for the average detection delay:
\begin{equation}
    \textnormal{ADD}(\Pi^\textnormal{A}) \le \frac{|\ln \alpha|}{d} - \frac{D \sqrt{\delta}}{d^{3/2} \sqrt{\chi_{2,1}}} \sqrt{|\ln \alpha|} + O(1)
    .
\end{equation}
\end{theorem}
\begin{proof}
\textbf{Step 1: Constraint Satisfaction and the Active Regime.} 
The satisfaction of PFA constraint follows by the property of Shiryaev-type policies. The ECB constraint is satisfied since $\beta_\alpha^* \triangleq \min(1, \tilde{\beta}_\alpha)\le\tilde{\beta}_\alpha$ and thus by~\eqref{eq:ECB_upper_relaxed}, for all $\alpha\in~]0,\min(1/2,1-\rho)[$ we have
\begin{align}
    \textnormal{ECB} \left( \tau^\textnormal{A} , \beta_\alpha^* \right)
    &\leq
    (\beta_\alpha^*)^2 \left( \frac{ \chi_{2,0}}{\rho} + \chi_{2,1} \textnormal{ADD}_\textnormal{relaxed} \right)
    \\&\leq
    (\tilde{\beta}_\alpha)^2 \left( \frac{ \chi_{2,0}}{\rho} + \chi_{2,1} \textnormal{ADD}_\textnormal{relaxed} \right)
    \triangleq \delta
    .
\end{align}

As $\alpha \to 0^+$, $|\ln \alpha| \to \infty$. Consequently, the denominator of $\tilde{\beta}_\alpha$ grows without bound, implying $\tilde{\beta}_\alpha \to 0^+$. There exists a critical threshold $\alpha_0 > 0$ such that for all $\alpha \in~]0, \alpha_0[~$, $\tilde{\beta}_\alpha < 1$. Thus, at least for this range of $\alpha$, which we call the active regime, we have $\beta_\alpha^* = \tilde{\beta}_\alpha$. 

\textbf{Step 2: Asymptotic Expansion of $\beta_\alpha^*$.}
Let $L_\alpha\triangleq |\ln\alpha|$. In the active regime,
$\beta_\alpha^*=\tilde{\beta}_\alpha$. Since
$|\ln(1-\alpha)|=o(1)$ as $\alpha\to0^+$, we have
\begin{align}
\textnormal{ADD}_\textnormal{relaxed}
&=
\frac{L_\alpha}{d}+O(1).
\end{align}
Consequently,
\begin{align}
\frac{\chi_{2,0}}{\rho}
+
\chi_{2,1}\textnormal{ADD}_\textnormal{relaxed}
&=
\frac{\chi_{2,1}}{d}L_\alpha+O(1)
\\
&=
\frac{\chi_{2,1}}{d}L_\alpha
\left(1+O(L_\alpha^{-1})\right).
\end{align}
Therefore,
\begin{align}
\beta_\alpha^*
&=
\sqrt{
\frac{\delta}{
\frac{\chi_{2,1}}{d}L_\alpha
\left(1+O(L_\alpha^{-1})\right)
}
}
\\
&=
\frac{\sqrt{\delta d}}{\sqrt{\chi_{2,1}}}
L_\alpha^{-1/2}
\left(1+O(L_\alpha^{-1})\right)^{-1/2}
\\
&=
\frac{\sqrt{\delta d}}{\sqrt{\chi_{2,1}}}
L_\alpha^{-1/2}
+
O(L_\alpha^{-3/2}).
\end{align}
In particular,
\begin{equation}
    \beta_\alpha^*=\Theta(L_\alpha^{-1/2}).
\end{equation}

\textbf{Step 3: Evaluating the Non-Asymptotic ADD Bound.}
Substituting $\beta_\alpha^*$ into the non-asymptotic ADD bound in
Theorem~\ref{thm:add_bound}, we have
\begin{align}
\textnormal{ADD}(\Pi^\textnormal{A})
&\le
\frac{
L_\alpha-|\ln(1-\alpha)|+C_\rho
}{
d+D\beta_\alpha^*
}
+
\frac{
\beta_\alpha^* V+2d\beta_\alpha^*D+d^2
}{
(d+D\beta_\alpha^*)^2
}.
\end{align}
Since $|\ln(1-\alpha)|=o(1)$, $C_\rho=O(1)$, and
$d+D\beta_\alpha^*\ge d>0$, the contribution of
$-|\ln(1-\alpha)|+C_\rho$ to the first fraction is $O(1)$.
The second fraction is also $O(1)$ because its numerator is $O(1)$
and its denominator is bounded below by $d^2$.
Hence,
\begin{align}
\textnormal{ADD}(\Pi^\textnormal{A})
&\le
\frac{L_\alpha}{d+D\beta_\alpha^*}
+
O(1).
\end{align}
Factoring out $d$ from the denominator gives
\begin{align}
\frac{L_\alpha}{d+D\beta_\alpha^*}
&=
\frac{L_\alpha}{d}
\left(
1+\frac{D}{d}\beta_\alpha^*
\right)^{-1}.
\end{align}
Since $\beta_\alpha^*\to0$, the Taylor expansion
$(1+x)^{-1}=1-x+O(x^2)$ yields
\begin{align}
\frac{L_\alpha}{d+D\beta_\alpha^*}
&=
\frac{L_\alpha}{d}
-
\frac{D}{d^2}L_\alpha\beta_\alpha^*
+
L_\alpha O((\beta_\alpha^*)^2).
\end{align}
Because $\beta_\alpha^*=\Theta(L_\alpha^{-1/2})$, we have
$L_\alpha O((\beta_\alpha^*)^2)=O(1)$. Therefore,
\begin{align}
\textnormal{ADD}(\Pi^\textnormal{A})
&\le
\frac{L_\alpha}{d}
-
\frac{D}{d^2}L_\alpha\beta_\alpha^*
+
O(1).
\end{align}

\textbf{Step 4: Final Substitution.}
Using the expansion of $\beta_\alpha^*$ from Step 2,
\begin{align}
-\frac{D}{d^2}L_\alpha\beta_\alpha^*
&=
-\frac{D}{d^2}L_\alpha
\left(
\frac{\sqrt{\delta d}}{\sqrt{\chi_{2,1}}}
L_\alpha^{-1/2}
+
O(L_\alpha^{-3/2})
\right)
\\
&=
-\frac{D\sqrt{\delta}}{d^{3/2}\sqrt{\chi_{2,1}}}
\sqrt{L_\alpha}
+
O(L_\alpha^{-1/2}).
\end{align}
Substituting this into the preceding bound and absorbing
$O(L_\alpha^{-1/2})$ into $O(1)$ yields
\begin{align}
\textnormal{ADD}(\Pi^\textnormal{A})
&\le
\frac{|\ln\alpha|}{d}
-
\frac{D\sqrt{\delta}}{d^{3/2}\sqrt{\chi_{2,1}}}
\sqrt{|\ln\alpha|}
+
O(1),
\end{align}
which completes the proof.
\color{black}

\end{proof}

\begin{remark}
Since the second-order asymptotics of the proposed achievability scheme matches that of the converse up to $O(1)$, it is immediate that for all $\delta>0$, as $\alpha\to0^+$, the optimal ADD of any causal policy also scales in terms of $|\ln\alpha|$ as
\begin{align}
    \textnormal{ADD}^\star(\alpha,\delta) = \frac{|\ln \alpha|}{d} - \frac{D \sqrt{\delta}}{d^{3/2} \sqrt{\chi_{2,1}}} \sqrt{|\ln \alpha|} + O(1)
    .
\end{align}
\end{remark}
\begin{remark}
In the asymptotic regime as $\alpha \to 0^+$, intuitively the ECB is overwhelmingly dominated by the post-change detection period. Consequently, the pre-change $\chi_{2,0}$ divergence becomes immaterial to the optimal scaling laws.
\end{remark}

\section{Numerical Results}
\label{sec:numer_resul}

\begin{figure}[!b]
    \centering
    \begin{subfigure}[b]{0.48\textwidth}
        \centering
        \includegraphics[width=\textwidth]{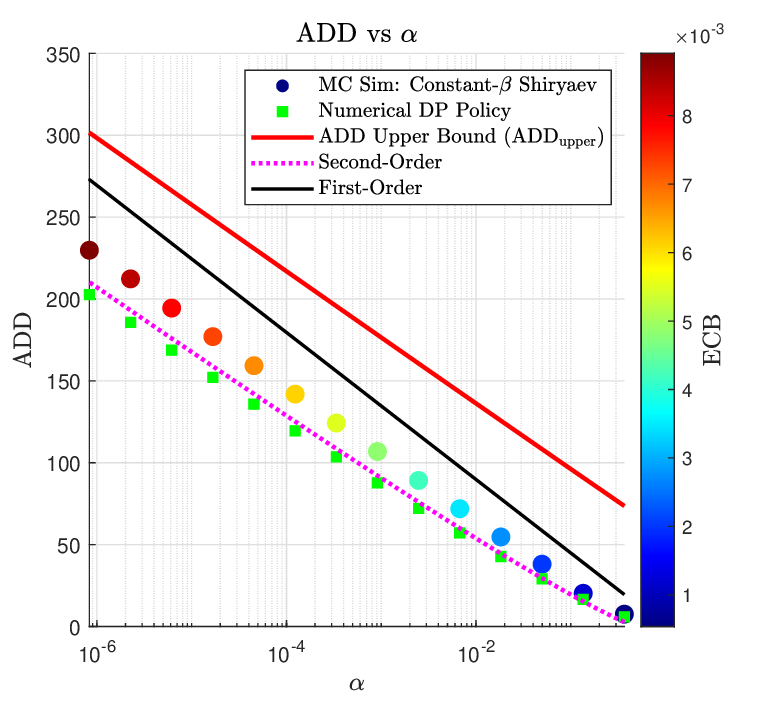}
        \caption{ADD versus $\alpha$ on a logarithmic scale.}
        \label{fig:raw_scale}
    \end{subfigure}
    \hfill
    \begin{subfigure}[b]{0.48\textwidth}
        \centering
        \includegraphics[width=\textwidth]{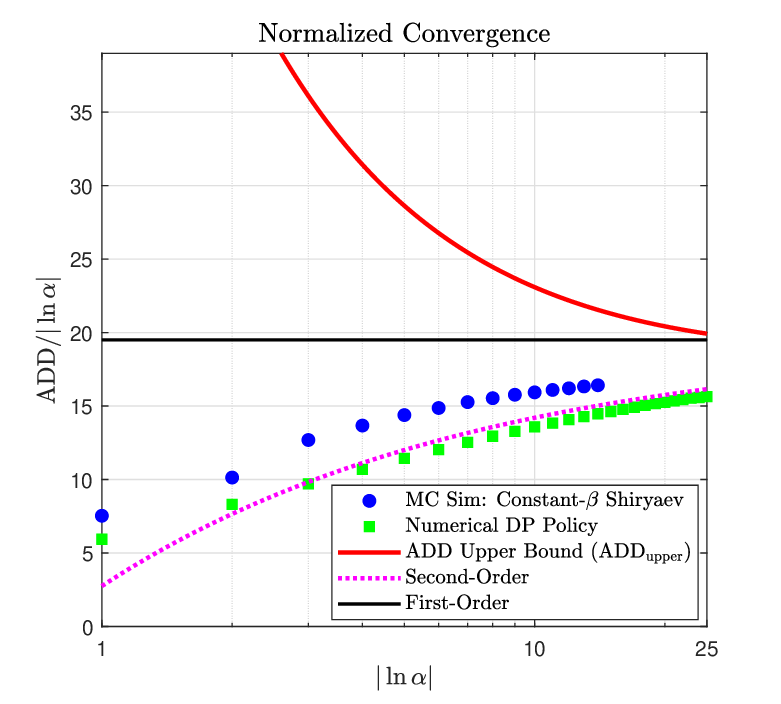}
        \caption{Normalized asymptotic convergence.}
        \label{fig:normalized_convergence}
    \end{subfigure}
    
    \caption{Performance validation of the proposed constant-sensing-probability Shiryaev-type policy $\Pi^\textnormal{A}=(\tau^\textnormal{A},\beta_{\alpha}^*)$. The parameters are $\rho = 1/20$, $\delta = 1/24$, $D\approx 0.8318$, $V\approx 1.9218$, $\chi_{2,1}=16/21$ and $\chi_{2,0}=1/6$. \textbf{(a)} average detection delay (ADD) plotted against the target PFA (probability of false alarm) constraint $\alpha$. \textbf{(b)} The ADD normalized by $|\ln \alpha|$, plotted against $|\ln \alpha|$.}
    \label{fig:asymptotic_validation}
\end{figure}

We now present numerical simulations to corroborate our theoretical findings. 
We use Monte Carlo simulation to estimate the ADD, PFA and ECB metrics of the proposed constant-sensing-probability Shiryaev-type policy $\Pi^\textnormal{A}=(\tau^\textnormal{A},\beta_{\alpha}^*)$ over the range of $\alpha$ values such that $ |\ln\alpha| \in\{1,2,\ldots,14\}$. The smallest $\alpha$ value is $\alpha_\textnormal{min}=\exp(-14)\approx 8.32\times10^{-7}$. We also compute an approximate dynamic programming (DP) policy by formulating Alice's policy design problem as a constrained POMDP~\cite{premkumar_optimal_2008,banerjee_bayesian_2011,geng_bayesian_2014} and applying discretized value iteration.

In Fig.~\ref{fig:asymptotic_validation}, we present the simulation result in a scenario where $\rho=1/20$ and $\delta = 1/24$. For Alice's pre-change (resp. post-change) observation channel, we choose $P^1_0=\textnormal{Ber}(0.2)$ and $P^1_1=\textnormal{Ber}(0.8)$, yielding $D\approx 0.8318$ and $V\approx 1.9218$. For Eve's pre-change (resp. post-change) observation channel under active or innocent probing, we choose $Q^1_0=\textnormal{Ber}(0.6)$ and $Q^0_0=\textnormal{Ber}(0.4)$ (resp. $Q^1_1=\textnormal{Ber}(0.7)$, $Q^0_1=\textnormal{Ber}(0.3)$), resulting in $\chi_{2,1}=16/21$ and $\chi_{2,0}=1/6$. 

Fig.~\ref{fig:raw_scale} plots the ADD against the PFA constraint $\alpha$ (in $\ln_{10}$ scale), and Fig.~\ref{fig:normalized_convergence} plots the normalized ADD, namely, $\textnormal{ADD}/|\ln \alpha|$, against $|\ln\alpha|$ (in $\ln_{10}$ scale). For both figures, we plot the Monte Carlo simulation result, the ADD upper bound via Lorden's theorem ($\textnormal{ADD}_\textnormal{upper}$), the first-order asymptotic of $\textnormal{ADD}\approx |\ln\alpha|/d$ (which represents the innocent policy, or no probing at all), and the second-order asymptotic $\textnormal{ADD}\approx \frac{|\ln \alpha|}{d} - \frac{D \sqrt{\delta}}{d^{3/2} \sqrt{\chi_{2,1}}} \sqrt{|\ln \alpha|}$. 

In Fig.~\ref{fig:raw_scale}, the simulated ECB for each data point is shown via a colorbar to verify that the ECB constraint $\delta$ is satisfied. The simulated ADD values fall below both $\textnormal{ADD}_\textnormal{upper}$ and the first-order asymptotic. Thus, even when $\delta=1/24<\min\{\chi_{2,1},\chi_{2,0}\}=1/6$, so that no deterministic probe can be used, probabilistic probing still improves ADD. The numerical DP policy is also close to the second-order asymptotic.

Fig.~\ref{fig:normalized_convergence} illustrates the convergence of $\textnormal{ADD}_\textnormal{upper}$, the second-order asymptotic, and the numerical DP policy towards the first-order asymptotic as $\alpha\to0^+$. The simulated constant-$\beta$ Shiryaev-type policy also appears to converge towards the asymptotics.

Both figures demonstrate substantial gains from covert probing. Compared to the no-probing baseline, the constant-$\beta$ Shiryaev policy achieves an absolute ADD reduction ranging from roughly $10$ at $\alpha=e^{-1}$ to $50$ at $\alpha=e^{-14}$, corresponding to a relative reduction from about $63\%$ to $16\%$.

Although the theory establishes matching second-order asymptotics as $\alpha\to0^+$, non-asymptotic gains from covert probing already appear at moderate values of $\alpha$, suggesting practical relevance for quickest change detection applications.

\section*{Acknowledgment}
Parts of this document have received assistance from generative AI tools to aid in the composition; the authors have reviewed and edited the content as needed and take full responsibility for it.

\bibliographystyle{IEEEtran}
\bibliography{references}

@Article{Huang_2020,
  author  = {Ke-Wen Huang and Hui-Ming Wang and Don Towsley and H. Vincent Poor},
  journal = {IEEE Transactions on Communications},
  title   = {{LPD} Communication: A Sequential Change-Point Detection Perspective},
  year    = {2020},
  month   = apr,
  volume  = {68},
  number  = {4},
  pages   = {2474--2490},
}

@Article{Huang_2021,
  author  = {Ke-Wen Huang and Hui-Ming Wang and H. Vincent Poor},
  journal = {IEEE Transactions on Information Theory},
  title   = {On Covert Communication Against Sequential Change-Point Detection},
  year    = {2021},
  month   = nov,
  volume  = {67},
  number  = {11},
  pages   = {7285--7303},
}

@InProceedings{Ramtin_2024,
  author    = {A. Ramtin and Z. Hare and L. Kaplan and P. Nain and V. Veeravalli and D. Towsley},
  booktitle = {Proc. of IEEE Military Communications Conference},
  title     = {Quickest Change Detection in the Presence of Covert Adversaries},
  year      = {2024},
  month     = oct,
  pages     = {1--6},
}

@Article{Ramtin_2025,
  author  = {Amir Reza Ramtin and Philippe Nain and Don Towsley},
  journal = {IEEE Signal Processing Letters},
  title   = {Quickest Change Detection in Continuous-Time in Presence of a Covert Adversary},
  year    = {2025},
  volume  = {32},
  pages   = {4299--4303},
}

@Article{Gagatsos2019,  author    = {Gagatsos, Christos N. and Bash, Boulat A. and Datta, Animesh and Zhang, Zheshen and Guha, Saikat},  journal   = {Physical Review A},  title     = {Covert sensing using floodlight illumination},  year      = {2019},  month     = {Jun},  pages     = {062321},  volume    = {99},  issue     = {6},  numpages  = {15},  publisher = {American Physical Society}}

@InProceedings{Tahmasbi2020b,  author       = {Mehrdad Tahmasbi and Matthieu R Bloch},  booktitle    = {Proc. of IEEE International Symposium on Information Theory},  title        = {Active Covert Sensing},  year         = {2020},  address      = {Los Angeles, CA},  month        = jun,  pages        = {840-845}}

@InProceedings{Tahmasbi2021,  author       = {Mehrdad Tahmasbi and Boulat Bash and Saikat Guha and Matthieu R Bloch},  booktitle    = {Proc. of IEEE International Symposium on Information Theory},  title        = {Signaling for Covert Quantum Sensing},  year         = {2021},  month        = jul,  pages        = {1041-1045}}

@InProceedings{Bash2017,  author          = {Boulat A. Bash and Christos N. Gagatsos and Animesh Datta and Saikat Guha},  booktitle       = {Proc. of IEEE International Symposium on Information Theory},  title           = {Fundamental limits of quantum-secure covert optical sensing},  year            = {2017},  address         = {Aachen, Germany},  month           = jun,  pages           = {3210--3214}}

@InProceedings{Goeckel2017,
  author    = {Dennis Goeckel and Boulat A. Bash and Azadeh Sheikholeslami and Saikat Guha and Don Towsley},
  booktitle = {Proc. of Asilomar Conference on Signals, Systems and Computers},
  title     = {Covert Active Sensing of Linear Systems},
  year      = {2017},
  address   = {Pacific Grove, CA},
  month     = nov,
  pages     = {1692--1696},
}

@Article{Tahmasbi2020c,
  author  = {Mehrdad Tahmasbi and Matthieu R. Bloch},
  journal = {IEEE Journal on Selected Areas in Information Theory},
  title   = {On Covert Quantum Sensing and the Benefits of Entanglement},
  year    = {2021},
  month   = mar,
  volume  = {2},
  number  = {1},
  pages   = {352--365},
}

@Article{Hao2022Demonstration,
  author  = {Shuhong Hao and Haowei Shi and Christos N. Gagatsos and Mayank Mishra and Boulat Bash and Ivan Djordjevic and Saikat Guha and Quntao Zhuang and Zheshen Zhang},
  journal = {Physical Review Letters},
  title   = {Demonstration of Entanglement-Enhanced Covert Sensing},
  year    = {2022},
  month   = jun,
  volume  = {129},
  number  = {1},
  pages   = {010501},
}

@InProceedings{Wang2024Covert,  author       = {Shi-Yuan Wang and Meng-Che Chang and Matthieu R. Bloch},  booktitle    = {Proc. of 58th Annual Conference on Information Sciences and Systems},  title        = {Covert Joint Communication and Sensing Under Variational Distance Constraint},  year         = {2024},  address      = {Princeton, NJ},  month        = mar}

@Article{Ker2007,  author       = {Ker, A.D.},  journal      = {IEEE Signal Processing Letters},  title        = {A Capacity Result for Batch Steganography},  year         = {2007},  number       = {8},  pages        = {525--528},  volume       = {14}}

@Article{Bash2013,  author       = {Bash, B.A and Goeckel, D. and Towsley, D.},  journal      = {IEEE Journal on Selected Areas in Communications},  title        = {Limits of Reliable Communication with Low Probability of Detection on {AWGN} Channels},  year         = {2013},  issn         = {0733-8716},  month        = {September},  number       = {9},  pages        = {1921-1930},  volume       = {31}}

@Article{Bloch2015b,  author       = {Matthieu R. Bloch},  journal      = {IEEE Transactions on Information Theory},  title        = {Covert Communication over Noisy Channels: A Resolvability Perspective},  year         = {2016},  issn         = {0018-9448},  month        = may,  number       = {5},  pages        = {2334-2354},  volume       = {62}}

@Article{Wang2016b,  author       = {L. Wang and G. W. Wornell and L. Zheng},  journal      = {IEEE Transactions on Information Theory},  title        = {Fundamental Limits of Communication With Low Probability of Detection},  year         = {2016},  issn         = {0018-9448},  month        = jun,  number       = {6},  pages        = {3493--3503},  volume       = {62}}

@InProceedings{Bai2014,  author       = {Cheng-Zong Bai and Gupta, V.},  booktitle    = {Proc. of American Control Conference},  title        = {On {Kalman} filtering in the presence of a compromised sensor: Fundamental performance bounds},  year         = {2014},  address      = {Portland, OR},  month        = {June},  pages        = {3029--3034}}

@InProceedings{Bai2015,  author       = {Cheng-Zong Bai and Pasqualetti, F. and Gupta, V.},  booktitle    = {Proc. of American Control Conference},  title        = {Security in stochastic control systems: Fundamental limitations and performance bounds},  year         = {2015},  address      = {Chicago, IL},  month        = {July},  pages        = {195--200}}

@Article{Moulin2003,  author       = {Moulin, P. and O'Sullivan, J.A.},  journal      = {IEEE Transactions on Information Theory},  title        = {Information-theoretic analysis of information hiding},  year         = {2003},  number       = {3},  pages        = {563--593},  volume       = {49}}

@Article{Wang2008,  author       = {Ying Wang and Moulin, P.},  journal      = {IEEE Transactions on Information Theory},  title        = {Perfectly Secure Steganography: Capacity, Error Exponents, and Code Constructions},  year         = {2008},  number       = {6},  pages        = {2706--2722},  volume       = {54}}

@article{Veeravalli_2024, title={Quickest Change Detection With Controlled Sensing}, volume={5}, ISSN={2641-8770}, journal={IEEE Journal on Selected Areas in Information Theory}, publisher={Institute of Electrical and Electronics Engineers (IEEE)}, author={Veeravalli, Venugopal V. and Fellouris, Georgios and Moustakides, George V.}, year={2024}, pages={1–11} }

@article{tartakovsky_general_2005,
	title = {General {Asymptotic} {Bayesian} {Theory} of {Quickest} {Change} {Detection}},
	volume = {49},
	number = {3},
	journal = {Theory of Probability \& Its Applications},
	publisher = {Society for Industrial and Applied Mathematics},
	author = {Tartakovsky, A. G. and Veeravalli, V. V.},
	month = jan,
	year = {2005},
	pages = {458--497},
}

@article{wald_cumulative_1944,
	title = {On {Cumulative} {Sums} of {Random} {Variables}},
	volume = {15},
	number = {3},
	journal = {The Annals of Mathematical Statistics},
	publisher = {Institute of Mathematical Statistics},
	author = {Wald, Abraham},
	month = sep,
	year = {1944},
	pages = {283--296},
}

@article{lorden_excess_1970,
	title = {On {Excess} {Over} the {Boundary}},
	volume = {41},
	number = {2},
	journal = {The Annals of Mathematical Statistics},
	publisher = {Institute of Mathematical Statistics},
	author = {Lorden, Gary},
	month = apr,
	year = {1970},
	pages = {520--527},
}

@article{shiryaev1963optimum,
  title={On optimum methods in quickest detection problems},
  author={Shiryaev, Albert N},
  journal={Theory of Probability \& Its Applications},
  volume={8},
  number={1},
  pages={22--46},
  year={1963},
  publisher={SIAM}
}

@inproceedings{banerjee_bayesian_2011,
	title = {Bayesian quickest change detection under energy constraints},
	booktitle = {2011 {Information} {Theory} and {Applications} {Workshop}},
	author = {Banerjee, Taposh and Veeravalli, Venugopal V.},
	month = feb,
	year = {2011},
	pages = {1--10},
}

@inproceedings{premkumar_optimal_2008,
	title = {Optimal {Sleep}-{Wake} {Scheduling} for {Quickest} {Intrusion} {Detection} {Using} {Wireless} {Sensor} {Networks}},
	issn = {0743-166X},
	booktitle = {{IEEE} {INFOCOM} 2008 - {The} 27th {Conference} on {Computer} {Communications}},
	author = {Premkumar, K. and Kumar, A.},
	month = apr,
	year = {2008},
	pages = {1400--1408},
}

@article{xie2021sequential,
  title={Sequential (quickest) change detection: Classical results and new directions},
  author={Xie, Liyan and Zou, Shaofeng and Xie, Yao and Veeravalli, Venugopal V},
  journal={IEEE Journal on Selected Areas in Information Theory},
  volume={2},
  number={2},
  pages={494--514},
  year={2021},
  publisher={IEEE}
}

@book{poor_quickest_2008,
	address = {Cambridge},
	title = {Quickest {Detection}},
	publisher = {Cambridge University Press},
	author = {Poor, H. Vincent and Hadjiliadis, Olympia},
	year = {2008},
}

@book{tartakovsky_sequential_2014,
	address = {New York},
	title = {Sequential {Analysis}: {Hypothesis} {Testing} and {Changepoint} {Detection}},
	publisher = {Chapman and Hall/CRC},
	author = {Tartakovsky, Alexander and Nikiforov, Igor and Basseville, Michele},
	month = aug,
	year = {2014},
}

@book{polyanskiy2025information,
  title={Information theory: From coding to learning},
  author={Polyanskiy, Yury and Wu, Yihong},
  year={2025},
  publisher={Cambridge university press}
}

@article{geng_bayesian_2014,
	title = {Bayesian {Quickest} {Change}-{Point} {Detection} {With} {Sampling} {Right} {Constraints}},
	volume = {60},
	number = {10},
	journal = {IEEE Transactions on Information Theory},
	author = {Geng, Jun and Bayraktar, Erhan and Lai, Lifeng},
	month = oct,
	year = {2014},
	pages = {6474--6490},
}

@inproceedings{chang_covert_2021,
	title = {Covert {Sequential} {Hypothesis} {Testing}},
	booktitle = {2021 {IEEE} {Information} {Theory} {Workshop} ({ITW})},
	author = {Chang, Meng-Che and Bloch, Matthieu R.},
	month = oct,
	year = {2021},
	pages = {1--6},
}

@inproceedings{rao_extensible_2022,
	title = {An {Extensible} {Covert} {Communication} {Scheme} {Over} the {AWGN} {Channel} {With} {Feedback}},
	booktitle = {2022 {IEEE} {Information} {Theory} {Workshop} ({ITW})},
	author = {Rao, Hangmei and Wang, Ligong},
	month = nov,
	year = {2022},
	pages = {368--373},
}

@incollection{veeravalli_quickest_2014,
	title = {Quickest {Change} {Detection}},
	volume = {3},
	booktitle = {Academic {Press} {Library} in {Signal} {Processing}},
	publisher = {Elsevier},
	author = {Veeravalli, Venugopal V. and Banerjee, Taposh},
	year = {2014},
	pages = {209--255},
}

\end{document}